\documentclass[11pt, letterpaper]{article}
\usepackage{setspace}
\usepackage{amsmath,mathrsfs,amsthm,dsfont}
\usepackage{caption}
\usepackage{verbatim}
\usepackage{subcaption}
\usepackage{graphicx,psfrag,epsf}
\usepackage{subcaption}
\usepackage{enumerate}
\usepackage{natbib}
\usepackage{bbm}
\usepackage{bm}
\usepackage{amssymb}
\usepackage[affil-it]{authblk}
\usepackage{verbatim}
\usepackage{extarrows}
\usepackage{framed}
\usepackage[colorlinks=true,citecolor=blue]{hyperref}
\usepackage{algorithm}
\usepackage{algorithmic}
\usepackage{url} 
\usepackage{float}
\usepackage{geometry}
\usepackage{lscape} 
\usepackage{multirow}
\usepackage{booktabs}
\usepackage{pifont}
 \usepackage{graphicx}
 \usepackage{pdfpages}
\usepackage{empheq}
\usepackage{thmtools}       
\usepackage[title]{appendix}
\usepackage[normalem]{ulem}
\usepackage{etoc}
\usepackage{enumitem}
\newcommand{\stkout}[1]{\ifmmode\text{\sout{\ensuremath{#1}}}\else\sout{#1}\fi}
\usepackage{amssymb}

\DeclareMathOperator*{\arginf}{arg\,inf}

\newcommand{\bfsym}[1]{\ensuremath{\boldsymbol{#1}}}

\def\bbeta{\bfsym \beta}
\def\btheta{\bfsym \theta}
\def\bgamma{\bfsym \gamma}

\def\btau{\bfsym {\tau}}

\def\bbeta{\boldsymbol{\beta}}

\def\bgamma{\boldsymbol{\gamma}}

\def\bA{\boldsymbol{A}}

\def\bY{\boldsymbol{Y}}

\def\bX{\boldsymbol{X}}
\def\bx{\boldsymbol{x}}

\def\ba{\boldsymbol{a}}

\def\bg{\boldsymbol{g}}

\newcommand{\piast}{\pi^\ast}
\newcommand{\pihat}{\hat{\pi}}

\newcommand\indep{\protect\mathpalette{\protect\independenT}{\perp}}
\def\independenT#1#2{\mathrel{\rlap{$#1#2$}\mkern2mu{#1#2}}}

\newcommand{\E}{\ensuremath{\mathbb{E}}}
\newcommand{\Prob}{\ensuremath{\mathbb{P}}}

\newcommand{\Var}{\mathbb{V}}

\newcommand{\addDR}{\textsf{addDR}}

\newcommand{\addIPW}{\textsf{addIPW}}

\theoremstyle{definition}
\newtheorem{assumption}{Assumption}

\newtheorem{theorem}{Theorem}
\newtheorem{lemma}{Lemma}
\newtheorem{example}{Example}
\newtheorem{definition}{Definition}

\newtheorem{proposition}{Proposition}

\allowdisplaybreaks

\begin{document}

\def\spacingset#1{\renewcommand{\baselinestretch}%
{#1}\small\normalsize} \spacingset{1}


\spacingset{1}
\title{{\bf Individualized Policy Evaluation and Learning under Clustered Network Interference}\thanks{We thank Georgia Papadogeorgou, Davide Viviano, and anonymous reviewers of the Alexander and Diviya Magaro Peer Pre-Review Program for useful comments.}}

\author{Yi Zhang\thanks{Ph.D. Student, Department of Statistics, Harvard University. 1 Oxford Street, Cambridge MA 02138. Email: \href{mailto:yi_zhang@fas.harvard.edu}{ yi$\_$zhang@fas.harvard.edu}} \quad\quad Kosuke Imai\thanks{Professor, Department of Government and Department of Statistics, Harvard University.  1737 Cambridge Street,
      Institute for Quantitative Social Science, Cambridge MA 02138.
      Email: \href{mailto:imai@harvard.edu}{imai@harvard.edu} URL:
      \href{https://imai.fas.harvard.edu}{https://imai.fas.harvard.edu}}}

      
\date{\today}

\maketitle

\etocdepthtag.toc{mtchapter}
\etocsettagdepth{mtchapter}{subsection}
\etocsettagdepth{mtappendix}{none}

\bigskip
\begin{abstract}
Although there is now a large literature on policy evaluation and learning, much of the prior work assumes that the treatment assignment of one unit does not affect the outcome of another unit.
Unfortunately, ignoring interference can lead to biased policy evaluation and ineffective learned policies.
For example, treating influential individuals who have many friends can generate positive spillover effects, thereby improving the overall performance of an individualized treatment rule (ITR).
We consider the problem of evaluating and learning an optimal ITR under clustered network interference (also known as partial interference), where clusters of units are sampled from a population and units may influence one another within each cluster. 
Unlike previous methods that impose strong restrictions on spillover effects, such as anonymous interference, the proposed methodology only assumes a semiparametric structural model, where each unit's outcome is an additive function of individual treatments within the cluster.
Under this model, we propose an estimator that can be used to evaluate the empirical performance of an ITR.
We show that this estimator is substantially more efficient than the standard inverse probability weighting estimator, which does not impose any assumption about spillover effects.
We derive the finite-sample regret bound for a learned ITR, showing that the use of our efficient evaluation estimator leads to the improved performance of learned policies.
We consider both experimental and observational studies, and for the latter, we develop a doubly robust estimator that is semiparametrically efficient and yields an optimal regret bound.
Finally, we conduct simulation and empirical studies to illustrate the advantages of the proposed methodology. 


\end{abstract}

\newpage

\section{Introduction}
\label{sec:intro}

During the past decade, a number of scholars have studied the problem of developing optimal individualized treatment rules (ITRs) that maximize the average outcome in a target population \citep[e.g.,][]{imai:stra:11,zhang2012estimating,zhao2012estimating,swaminathan2015counterfactual,kitagawa2018should,athey2021policy,zhang2022safe}. 
Beyond academia, these methods have played an essential role in the implementation of personalized medicine and micro-targeting in advertising and political campaigns.
In addition, new methodologies have been developed to evaluate the empirical performance of learned ITRs \citep{imai:li:23}.

Much of this existing policy evaluation and learning literature assumes no interference between units, i.e., one's outcome is not affected by the treatments of others.
Yet, in real-world applications, such spillover effects are the norm rather than an exception.
This means that there is a potential to exploit spillover effects when learning an optimal ITR by incorporating information about individuals and their network relationships.
For example, assigning influential and well-connected students to an anti-bullying program may more effectively reduce the number of conflicts within a school \citep{paluck2016changing}.
Another example is that individuals who are at the center of the social network can spread information more widely \citep{banerjee2019using}.

Despite these potential advantages, there exist key methodological challenges when learning ITRs in the presence of interference between units.
First, the structure of spillover effects is often little understood.
It is difficult to obtain information about people's relationships, and the existence of unobserved networks can invalidate the performance evaluation of the learned ITRs \citep{egami_2021}.
Second, the total number of possible treatment allocations increases exponentially, leading to the difficulty of inferring causal effects of high-dimensional treatments.

Thus, efficient individualized policy learning and evaluation require an assumption that is sufficiently informative to constrain the structure of spillover effects.
At the same time, we must avoid unrealistic assumptions.
In particular, many researchers assume anonymous (stratified) interference, in which spillover effects are determined by the number of treated neighbors regardless of which neighbors are treated \citep{toward_hudgens_2008,large_liu_2014,viviano2019policy}.
However, the way in which one unit's treatment influences another unit's outcome often depends on their specific relationship. 
Our goal is to relax this anonymous interference assumption by allowing for heterogeneous spillover effects for efficient policy evaluation and learning.

In this paper, we consider the evaluation and learning of individualized policies under clustered (partial) network interference \citep[see Section~\ref{sec:setup};][]{sobel2006randomized,toward_hudgens_2008,eric2012}.
Under this setting, units are grouped into non-overlapping clusters and interference arises within each of these disjoint clusters rather than between them.
In other words, the outcome of one unit is possibly affected by the treatments of the other units in the same cluster, but not by those of units in other clusters.  We first focus on experimental studies where the treatment assignment probabilities are known (Sections~\ref{sec:method}~and~\ref{sec:empiricalPL}).  We then consider observational studies, in which propensity scores are unknown and must be estimated (Section~\ref{sec:extension:obs}).

We propose a methodology for the evaluation and learning of individualized policies based on a semiparametric structural model.
Specifically, we assume that each individual's conditional mean outcome function is additive in the treatment vector of all individuals within the same cluster.
Importantly, under this additivity assumption, the proposed model uses individual-specific nonparametric functions that place no restriction on the heterogeneity of spillover effects.
The model, for example, accommodates the possibility that well-connected units within a cluster have a greater influence on other units.
Indeed, this semiparametric model contains as a special case the standard parametric model based on the anonymous (stratified) interference assumption.

Next, we introduce a new policy evaluation estimator that exploits the proposed semiparametric model.
Our estimator, which we call the additive inverse-probability weighting (\addIPW) estimator, leverages the semiparametric structural assumption of spillover effects without the need to fit an outcome model.
We show that the \addIPW{} estimator is unbiased and is more efficient than the standard IPW estimator, which makes no assumption about the structure of within-cluster spillover effects \citep{toward_hudgens_2008,liu2016inverse,tchetgen2012causal}.
Finally, using this \addIPW{} estimator, we find an optimal ITR within a pre-specified policy class.
Following the previous works \citep{kitagawa2018should,athey2021policy,zhou2023offline}, we show that the empirical policy optimization problem can be formulated as a mixed-integer linear program, which can be solved with off-the-shelf optimization tools.
Theoretically, we establish a finite-sample regret bound and show that the regret converges at the same optimal rate as the standard i.i.d. policy learning settings.

We further extend our methodology to observational studies with unknown treatment assignment probabilities.
We introduce an efficient semiparametric doubly robust policy evaluation estimator under our additivity assumption. This \addDR{} estimator is robust to estimation errors, either in the propensity score or outcome models. Leveraging results from the semiparametric literature, we establish an asymptotic regret bound that achieves a convergence rate comparable to that in the experimental setting.

We conduct simulation studies to examine the finite-sample empirical performance of the proposed methodology (Section~\ref{sec:simulation}).
We demonstrate the superiority of the proposed methodology over the optimal policy learned with the standard IPW estimators.
Lastly, we apply our methodology and learn an optimal ITR to increase the household-level school attendance among Colombian schoolchildren through its conditional cash transfer program (Section~\ref{sec:app}).

\paragraph{Related work.}
Numerous scholars have studied the problem of interference between units \citep[e.g.,][]{liu2016inverse,aronow2017estimating,bass:fell:17,athey2018exact,leung2020treatment,imai2021interference,savje2021average,hu2022average,li2022random,puelz2022,gao2023causal,ambarish2023}.
Much of this literature, however, has focused upon the estimation of various causal effects, including spillover and diffusion effects.
At the same time, others have studied policy learning and evaluation \citep[e.g.,][]{zhang2012estimating,zhao2012estimating,swaminathan2015counterfactual,kallus2018balanced,kitagawa2018should,chernozhukov2019semi,athey2021policy,jin2022policy,imai:li:23,zhou2023offline}.
The vast majority of this policy learning literature requires the assumption of no interference between units.
In contrast to these existing works, we consider the problem of policy learning and evaluation under clustered network interference where units influence one another within each cluster.

A relatively small number of studies have addressed the challenge of interference when learning optimal ITRs.
Some utilize parametric outcome models \citep{kitagawa2023should,ananth2020optimal}, while others adopt an exposure mapping approach \citep{ananth2020optimal,viviano2019policy,park2023minimum}.
These works often impose strong functional form assumptions on the conditional mean outcome model.
In particular, a vast majority of previous studies, if not all, rely on anonymous (or stratified) interference where spillover effects are assumed to be a function of the number of treated units in a cluster \citep{viviano2019policy,ananth2020optimal,park2023minimum}.
Our approach avoids placing these restrictive assumptions on the structure of spillover effects.

Two of the aforementioned studies are closely related to our work.
First, \cite{viviano2019policy} assumes anonymous interference but is able to develop an optimal policy learning methodology under a single network setting.
Our methodology avoids the anonymous interference assumption, but requires a random sample of clusters from a target population.
Second, like our work, \cite{park2023minimum} study policy learning in clustered network interference settings.
The authors consider an optimal cluster-level treatment policy that suggests the minimum proportion of treated units required within a cluster to achieve a pre-defined target average outcome level. 
Their methodology, however, cannot specify which individual units within a cluster should receive treatment.
In contrast, we propose an individualized policy learning methodology that optimally assigns treatments to individuals based on both individual and network characteristics.

There also exists a literature on policy evaluation under clustered network interference.
For example, \cite{eric2012} and \cite{large_liu_2014} study causal effect estimates under a ``type-B" policy, where units independently select to receive the treatment with a uniform probability.
In addition, \cite{papadogeorgou2019causal} and \cite{barkley2020causal} propose policy-relevant causal estimands based on a shift in parametric propensity score models. \cite{lee2022efficient} introduce an incremental propensity score intervention that further relaxes these parametric assumptions. 
In contrast, we propose an efficient policy evaluation estimator by leveraging the semiparametric additivity assumption that places a relatively weak restriction on the structure of spillover effects.
Our semiparametric model is closer to the one considered by \cite{yu2022estimating} who use the model to estimate treatment effects in a design-based single network setting.

As discussed later, a fundamental challenge of policy learning and evaluation under clustered network intereference is that the treatment assignment is high-dimensional.
In particular, the number of possible treatment combinations grow exponentially as the cluster size increases.
The problem of policy learning and evaluation with high-dimensional treatments has been studied in different contexts.
For example, \cite{Xu2023optimal} examine policy learning with multiple treatments, while \cite{chernozhukov2019semi} study policy learning with continuous treatments.

\section{The Problem Statement}\label{sec:setup}

We begin by describing the setup and notation of individualized policy learning under clustered network interference. 

\subsection{Setup and Notation}

Consider a setting in which observed units can be partitioned into clusters, such as households, classrooms, or villages. 
Let $M_i$ denote the number of units in cluster $i \in$ $\{1,2, \ldots, n\}$.
For unit $j \in \{1,2, \ldots, M_i\}$ in cluster $i$, we let $Y_{ij}\in \mathbb{R}$ represent the observed outcome, $A_{ij}\in\{0,1\}$ denote the binary treatment condition assigned to this unit, and $\bX_{ij}\in \mathbb{R}^p$ be the vector of $p$-dimensional pre-treatment covariates.
We use $\bY_i=\left(Y_{i 1}, Y_{i 2}, \ldots, Y_{i M_i}\right)^\top $, $\bA_{i}=\left(A_{i 1}, A_{i 2}, \ldots, A_{i M_i}\right)^\top $ and $\bX_i=\left(\bX_{i1}, \bX_{i2}, \ldots, \bX_{i M_i}\right)^\top$ to denote the cluster-level vectors of outcome and treatment, and the cluster-level matrix of pre-treatment covariates, respectively.
Lastly, we use $\mathcal{A}$, $\mathcal{X}$, and $\mathcal{Y}$ to represent the support of $A_{ij}$, $\bX_{ij}$, and $Y_{ij}$.
Thus, $\mathcal{A}\left(m_i\right)=\{0,1\}^{m_i}$  is the set of all the possible $2^{m_i}$ combinations of individual-level treatment assignments within a cluster of size $m_i$. 

Throughout the paper, we assume clustered network interference, also called partial interference, where interference between individuals only occur within a cluster rather than across different clusters.
This assumption is widely used in the literature when units are partitioned into non-overlapping clusters \citep[e.g.,][]{sobel2006randomized,toward_hudgens_2008,liu2016inverse,bass:fell:17,liu2019doubly,imai2021interference}.
Under this assumption, we define the potential outcome of one unit as a function of their own treatment as well as the treatment of others in the same cluster.
Formally, we let $Y_{ij}\left(\ba_i\right)$ denote the potential outcome of unit $j$ in cluster $i$ when the cluster is assigned to the treatment vector $\ba_i=\left(a_{i 1}, a_{i 2}, \ldots, a_{i M_i}\right)\in\mathcal{A}\left(M_i\right)$ and define the cluster-level potential outcome vector by
$\boldsymbol{Y}_i\left(\boldsymbol{a}_i\right)=\left(Y_{i1}\left(\boldsymbol{a}_i\right), Y_{i 2}\left(\boldsymbol{a}_i\right), \ldots, Y_{iM_i}\left(\boldsymbol{a}_i\right)\right)^\top$. 
Additionally, we make the standard consistency assumption that the vector of observed outcome, i.e., $\bY_i= \sum_{\ba_i \in \mathcal{A}\left(M_i\right)} \bY_i\left(\ba_i\right)\mathds{1}(\bA_i=\ba_i)$.

\begin{assumption}[i.i.d. Clusters]\label{ass:iidcluster}
    A total of $n$ clusters are sampled independently from a super-population.  That is, we assume $O_i=\left(\{\bY_i(\ba_i)\}_{\ba_i \in \mathcal{A}(M_i)}, \bA_i, M_i, \bX_i\right)$ is an independently and identically distributed random vector whose distribution is $\mathcal{O}$. For notational simplicity, we will include $M_i$ as part of $\bX_i$ unless explicitly noted otherwise.
\end{assumption}

Consider the cluster-level (generalized) propensity score under clustered network interference by jointly modelling all treatment assignments in a cluster, i.e., $e\left(\ba_i \mid \bX_i\right):=\Prob(\bA_i=\ba_i \mid \bX_i)$.
This represents the probability of observing treatment vector $\ba_i \in \mathcal{A}\left(M_i\right)$ given the cluster-level covariates $\bX_i\in \mathcal{X}\left(M_i\right)$, where $\mathcal{X}\left(M_i\right)$ is the support of $\bX_i$ for a cluster of size $M_i$.
The following assumption is maintained throughout this paper. 
\begin{assumption}[Strong Ignorability of Treatment Assignment]\label{ass:DGP} The following conditions hold for all $\ba_i \in \mathcal{A}\left(M_i\right)$ and $\bx_i \in \mathcal{X}\left(M_i\right)$.
\begin{enumerate}[label=(\alph*)]
\item \textit{Unconfoundedness}: $\bY_i\left(\ba_i\right) \indep \bA_i \mid \bX_i=\bx_i$ 
\item \textit{Positivity}: $\exists$ $\epsilon>0$ such that
 $\epsilon<e\left(\ba_i \mid \bx_i\right)<1-\epsilon$
\end{enumerate}
\end{assumption}
We first consider experimental settings, where propensity scores are known and hence Assumption~\ref{ass:DGP} is satisfied by design.
In Section~\ref{sec:extension:obs}, we extend our methodology to observational studies, in which the propensity scores are unknown and must be estimated under Assumption~\ref{ass:DGP}.

\subsection{Individualized Policy Learning Problem}

Without loss of generality, assume that a greater value of outcome is preferable.  Thus, we focus on learning an optimal individualized treatment rule (ITR) that maximizes the within-cluster average of outcome while allowing for possible interference between units within each cluster.
Although the individual outcome of interest is aggregated to the cluster-level, the learned policy will be a function of both individual-level and cluster-level characteristics.  Thus, our methodology provides guidance to policy makers about which individuals in what type of cluster should receive treatment.

We define an ITR as a mapping from the individual-level covariates $\bX_{ij}$ to a binary treatment decision for an individual unit, i.e., $\pi: \mathcal{X}\rightarrow\{0,1\}$.
The cluster-level vector of treatment assignment under a policy $\pi$, therefore, is given by 
$$\{\pi(\bX_{ij})\}_{j=1}^{M_i} \ =  \ (\pi(\bX_{i1}),\ldots,\pi(\bX_{iM_i}))\in\{0,1\}^{M_i}.$$
In practice, researchers can also specify a class of policies $\Pi$ that incorporate various constraints.
For example, the linear policy class is defined as,
\[
\Pi=\{\pi:\mathcal{X}\rightarrow \{0,1\}\mid \pi(X) = \mathds{1}(\gamma_0+X^\top\gamma\geq0),\quad \gamma_0\in\mathbb{R},\gamma\in\mathbb{R}^p\}.
\]
Other forms of policies, such as decision trees and decision tables, have also been considered \citep[e.g.,][]{athey2021policy,benm:etal:21,jia2023bayesian,zhou2023offline}.

Our individualized policy learning formulation enables different treatment decisions for units within the same cluster.
This potentially yields a much improved outcome and allows for individual-level cost, fairness, and other considerations.
Indeed, $\bX_{ij}$ may include not only the attributes of unit $j$, but also those of its neighbors or friends within the same cluster $i$, as well as cluster-level characteristics such as cluster size $M_i$ or other network attributes.
It is also possible to include the interactions between these individual, network, and cluster-level characteristics, as long as the dimension of $\bX_{ij}$ remains fixed.

While we allow policies to depend on any covariates at both individual and cluster levels, we impose the restriction that each individual's treatment decision does not directly depend on those of others.
Thus, individualization is achieved through covariates rather than formulating a different treatment rule $\pi$ for each individual.
In other words, our policy class precludes any joint treatment rules across individuals even when such policies achieve better outcomes.
Although a policy class that includes joint treatment rules is more general, specifying and optimizing such a policy when the cluster size varies leads to additional parameterization and optimization challenges.

To evaluate a policy $\pi\in\Pi$, we follow existing work and focus on the population mean of the potential outcome distribution.
A key departure from the literature on policy learning without interference is that we must consider how each individual's outcome depends on the treatment assignments of the other units in the same cluster.
Specifically, we define the value of policy $\pi$ as:
\begin{equation}\label{equ:V}
    V(\pi)=\mathbb{E}_{\mathcal{O}}\left[\frac{1}{M_i} \sum_{j=1}^{M_i} {Y}_{ij}\left(\{\pi(\bX_{ij})\}_{j=1}^{M_i}\right)\right],
\end{equation}
where the expectation is taken over the super-population of clusters $\mathcal{O}$.
Given a pre-specified policy class $\Pi$, we wish to find an optimal policy $\piast$ within this class that maximizes the policy value,
\begin{equation}\label{equ:piast}
    \begin{aligned}
      \pi^\ast \in \underset{\pi\in\Pi}{ \operatorname{argmax}}\   V(\pi).
        \end{aligned}
\end{equation}

As mentioned above, our framework does not allow for a policy class that directly constrains joint treatment decisions across individuals.
However, it is possible to discourage (or encourage) certain joint treatment decisions by incorporating a treatment cost function that depends on the treatment decisions of multiple individual units within the same cluster.
For example, in Section~\ref{sec:app}, we use a cost function that is proportional to the total number of units who receive treatment within each cluster.

Finally, we define the \textit{regret} of policy $\pi$ as the difference in value between the optimal policy $\piast$ and the policy under consideration:
\begin{equation}\label{equ:regret:def}
    R(\pi) = V\left(\pi^*\right)-V(\pi).
\end{equation}
Our goal is to learn a policy with minimal regret within a policy class whose complexity is bounded.

\section{Policy Evaluation}
\label{sec:method}

Before we turn to the problem of learning an optimal ITR, we must identify and estimate the policy value $V(\pi)$ defined in Equation~\eqref{equ:V} for any given policy $\pi$.
We first show that the high-dimensional nature of treatments under clustered network interference makes the standard IPW estimator inefficient.
To address this challenge, we propose a semiparametric model that imposes a constraint on the structure of spillover effects while allowing for unknown heterogeneity in how units affect one another in the same cluster.
We then propose a new efficient estimator of policy value that leverages this semiparametric outcome model.

\subsection{The Inverse-probability-weighting (IPW) Estimator}\label{sec:ident:naiveIPW}

Under Assumption~\ref{ass:DGP}, the policy value function in Equation~\eqref{equ:V} can be identified as
\begin{equation*}\label{equ:V:naiveidentify}
    V(\pi)=\mathbb{E}\left[\frac{1}{M_i} \sum_{j=1}^{M_i} \mathbb{E}\left[{Y}_{ij}\mid \bA_i=\{\pi(\bX_{ij})\}_{j=1}^{M_i},\bX_i\right]\right].
\end{equation*}
Then, the following IPW estimator can be used to estimate the policy value $V(\pi)$,
\begin{equation}\label{equ:Vhat:fullIPW}
    \begin{aligned}
      \widehat{V}^{\text{IPW}}(\pi)
      &=\frac{1}{n}\sum_{i=1}^{n} \frac{\mathds{1}\left\{  \bA_i=
      \{\pi(\bX_{ij})\}_{j=1}^{M_i}
     \right\}}{  e(
 \{\pi(\bX_{ij})\}_{j=1}^{M_i}
       \mid\bX_i)}\overline{Y}_{i},
    \end{aligned}
\end{equation}
where $\overline{Y}_i=\sum_{j=1}^{M_i} Y_{i j}/M_i$ is the cluster-level average of the observed individual outcomes.
The weight for cluster $i$ is given by the reciprocal of the cluster-level propensity score.

As noted earlier, we first consider an experimental study in which the propensity score is known.
Under this setting, it is straightforward to show that this IPW estimator is both unbiased and $\sqrt{n}-$consistent.
Similar IPW estimators have been applied to policy learning without interference \citep[e.g.,][]{swaminathan2015counterfactual,zhang2012estimating,zhao2012estimating,kitagawa2018should}.
Indeed, the IPW estimator given in Equation~\eqref{equ:Vhat:fullIPW} is a natural extension of the standard IPW estimator to clustered network interference and has been used for the estimation of treatment effects in the presence of interference \citep{tchetgen2012causal,liu2016inverse,papadogeorgou2019causal,imai2021interference}.

While Assumption~\ref{ass:DGP} is sufficient for causal identification, IPW estimators often suffer from a large variability when a data set is of moderate size.
This issue is exacerbated in clustered network interference settings because $\bA_i$ is a high-dimensional treatment vector.
Because there exist a large number of treatment combinations, the probability of any treatment assignment combination can take an extremely small value, leading to a high variance.
As an example, consider an Bernoulli randomization design where units are independently assigned to the treatment condition with probability $q\leq 0.5$.
Then, the inverse propensity score for treating all units in the same cluster is given by $1/q^{M_i}$, with the denominator scaling exponentially in cluster size $M_i$.
In practice, therefore, we may not even observe a single cluster whose treatment vector aligns with the treatment assignment under a given policy especially when the cluster size is relatively large.
In such cases, the IPW estimator is not applicable as one cannot precisely estimate the policy value.

A large variance of the standard IPW estimator will negatively affect the performance of subsequent policy learning, which optimizes the empirical estimate of the value $\widehat{V}^{\text{IPW}}(\pi)$ across all policies in the policy class $\Pi$.
A simulation study in Section~\ref{sec:simulation} demonstrates that a learned policy based on $\widehat{V}^{\text{IPW}}(\pi)$ can exhibit a slow rate of learning and substantially deviate from an optimal policy. 
An alternative is to use an efficient semiparametric estimator that is known to be asymptotically efficient and often exhibits improved finite-sample performance relative to the standard IPW estimator \citep{park2022efficient}.
Unfortunately, these efficient semiparametric estimators may still suffer from a large variance when the inverse-propensity weights are large.

Indeed, it is generally impossible to improve the IPW estimator without an additional assumption given that the IPW estimator has been shown to be minimax optimal (up to some constant factors) in the non-asymptotic regime \citep{wang2017optimal}.
This observation motivates our semiparametric modeling assumption, to which we now turn.

\subsection{A Semiparametric Additive Outcome Model}\label{sec:ident:PIEstor}

To reduce a large estimation variance, we propose a semiparametric model that places a restriction on the interference structure.
We seek a relatively weak but informative assumption that allows for a sufficiently complex pattern of interference within each cluster, while significantly improving the efficiency of policy value estimate.
Specifically, we assume that each individual's conditional mean potential outcome is linear in the treatment vector of all individuals within the same cluster.
\begin{assumption}[Heterogeneous Additive Outcome Model]\label{ass:additive}
The potential outcome model satisfies 
\begin{equation}\label{equ:additive}
\E\left[Y_{ij}(\ba_i)\mid \bX_i\right]=g^{(0)}_{j}(\bX_i)+\sum_{k=1}^{M_i}g^{(k)}_{j}(\bX_i)a_{ik} 
\end{equation}
for all $\ba_i$, where $g^{(k)}_{j}(\cdot):\mathcal{X}(m)\rightarrow \mathbb{R}$ and $\bg_{j}(\bx)=\left(g^{(0)}_{j}(\bx),g^{(1)}_{j}(\bx),\ldots,g^{(m)}_{j}(\bx)\right)^\top$ is an unknown treatment effect function sets that may vary across individual units with $m$ denoting a specific realization of the cluster size $M_i$. 
\end{assumption}

The expectation in Equation~\eqref{equ:additive} is taken over the sampling distribution of the clusters. The additive relationship holds for all units within each cluster and is invariant to any permutation of the unit index $j$ within cluster $i$.  The reason is that the coefficients $g_j^{(k)}(\bX_i)$ for $k=0,1,\dots,M_i$ are completely unrestricted and can be adapted to any ordering of units.
Finally, this assumption only characterizes the conditional expectation of potential outcomes given observed covariates $\bX_i$ while allowing for the presence of unmeasured effect modifiers that are not confounders. 

A key feature of the proposed model is that it does not restrict the degree of heterogeneity in spillover effects.
This is important because how individuals affect one another may depend on their specific relationships.
For example, the influence of one's close friend may be greater than that of an acquaintance. 
Moreover, spillover effects may be asymmetric with one person exerting greater effects on others without being influenced by them.
In other words, the causal effect of one unit's treatment on another unit's outcomes can depend on the characteristics of both units and their relationship.
Our model accommodates these and other possibilities by representing spillover effects with a nonparametric function that is specific to a directed relationship from one unit to another that depends on the whole cluster-level vector of characteristics.

The proposed model incorporates, as special cases, more restrictive assumptions on the structure of interference considered in the literature.
For example, scholars studied the following parametric linear-in-means model \citep[e.g.,][]{liu2016inverse,liu2019doubly,park2022efficient}. 
\begin{example}[Linear-in-means model]
\begin{equation}\label{equ:example1}
     \E\left[Y_{ij}(\ba_i)\mid \bX_i\right]=\gamma_1+\gamma_2 a_{ij}+ \gamma_3  \overline{a}_{i(-j)}+ \bgamma^\top_4\bX_{ij} + \bgamma^\top_5 \overline{\bX}_{i(-j)},
\end{equation}
where the potential outcome model is assumed to be a linear function of one's own treatment assignment and characteristics, the proportion of treated units (other than yourself) in the cluster $\overline{a}_{i(-j)}=\sum_{k\neq j}a_{ik}/(M_i-1)$, and the cluster-level mean of other units' characteristics $\overline{\bX}_{i(-j)}=\sum_{k\neq j}\bX_{ik}/(M_i-1)$, and $\gamma_1,\ldots,\bgamma_5$ are the coefficients.
We can show that this model is a special case of the proposed model given in Equation~\eqref{ass:additive} by setting $g^{(0)}_{j}(\bX_i)=\gamma_1+\bgamma^\top_4\bX_{ij} + \bgamma^\top_5 \overline{\bX}_{i(-j)}$, $g^{(j)}_{j}(\bX_i)=\gamma_2$, and $g^{(k)}_{j}(\bX_i)=\gamma_3(M_i-1)$ for all $k\neq j$.
\end{example}

A popular model based on the anonymous (stratified) interference assumption \citep[e.g.,][]{large_liu_2014,toward_hudgens_2008,tchetgen2012causal,bargagli2020heterogeneous,viviano2019policy,park2023minimum} is also a special case of our model.
\begin{example}[Additive nonparametric effect model under anonymous interference]
\begin{equation}\label{equ:example2}
     \E\left[Y_{ij}(\ba_i)\mid \bX_i\right]=h_0(\bX_{ij},\bX_{i(-j)})+h_1(\bX_{ij},\bX_{i(-j)})a_{ij}+h_2(\bX_{ij},\bX_{i(-j)})\overline{a}_{i(-j)},
\end{equation}
where $h:=(h_0,h_1,h_2)$ is a vector of nonparametric functions. 
This model is a nonparametric extension of the parametric model given in Equation~\eqref{equ:example1} under anonymous interference.
This model is a special case of Assumption~\eqref{ass:additive} with
$g^{(k)}_{j}(\bX_i)=h_2(\bX_{ij},\bX_{i(-j)})$ for all $k\neq j$.
\end{example}

While the proposed model enables units to arbitrarily influence one another within each cluster, it rules out an interaction between spillover effects.
For example, the effect of treating one child in a household cannot depend on whether their siblings are treated.
In Section~\ref{sec:higherorder}, we show that the proposed model can be extended to a more complex, semiparametric polynomial model that incorporates interaction terms.
Such a general model, however, may yield a highly variable estimate of policy value, worsening the performance of learned treatment rules.
In Section~\ref{sec:simulation}, we provide empirical evidence that the proposed model serves as a good approximation to a more complex interference structure and that the proposed estimator, introduced below, substantially outperforms the IPW estimator, which makes no structural assumption.

\subsection{Identification and Estimation}\label{sec:addIPW}

We now propose an estimator of policy value function that leverages the additive structural assumption (Assumption~\ref{ass:additive}).
Under this assumption, model complexity grows only linearly with cluster size $M_i$ rather than at an exponential rate.
Before describing our estimator, we introduce an additional assumption about the treatment assignment mechanism that replaces Assumption~\ref{ass:DGP}(b).
Specifically, we assume that treatments are assigned to individuals independently within each cluster conditional on the observed covariates.
\begin{assumption}\label{ass:factorCPS}(Factored cluster-level propensity score)
   The cluster-level treatment probability satisfies
   $$\Prob(\bA_i=\ba_i \mid \bX_i)=\prod_{j=1}^{M_i}\Prob(A_{ij}=a_{ij} \mid \bX_i).$$
In addition, there exists $\eta>0$ such that
 $\eta<\Prob(A_{ij}=a_{ij} \mid \bX_i)<1-\eta$ for any $a_{ij}\in\{0,1\}$.
\end{assumption}

Assumption~\ref{ass:factorCPS} is satisfied under a Bernoulli randomized trial, i.e., each $A_{ij}\sim \operatorname{Bern}(p_{ij})$ for $p_{ij}\in(0,1)$.
The assumption implies the conditional independence of treatment assignments across individuals within the same cluster \citep{eric2012}, i.e., $A_{ij}\indep\bA_{i(-j)}\mid \bX_i$ where $\bA_{i(-j)} \in \mathcal{A}\left(M_i-1\right)$ denotes the vector of treatment indicators for all units in cluster $i$ other than unit $j$.
Recall that Assumption~\ref{ass:DGP}(b) demands strong overlap for every treatment combination, which is unlikely to hold when $\bA_i$ is a high-dimensional treatment vector.
In contrast, Assumption~\ref{ass:factorCPS} only requires the individual-level propensity scores to be sufficiently bounded from zero. This means that the constant, which satisfies Assumption~\ref{ass:DGP}(b), is likely to satisfy Assumption~\ref{ass:factorCPS}.

Under this setup, we introduce the following additive Inverse-Propensity-Weighting (\addIPW) estimator of the policy value $V(\pi)$ for a given policy $\pi\in\Pi$,
\begin{equation}\label{equ:Vhat:addIPW}
  \widehat{V}^{\addIPW}(\pi)= \frac{1}{n}\sum_{i=1}^{n}\overline{Y}_{i} \left[\sum_{j=1}^{M_i}
  \left( \frac{\mathds{1}\left\{
       A_{ij}=\pi(\bX_{ij})
     \right\}}{  e_j(
 \pi(\bX_{ij})
      \mid \bX_{i})}
  -1\right)+1\right].  
\end{equation}
where $e_j\left(a_{ij} \mid \bX_i\right):=\Prob(A_{ij}=a_{ij} \mid \bX_i)$ is the individual-level propensity score, and subscript $j$ emphasizes the fact that units are allowed to have different propensity score models.

The \addIPW{} estimator is a weighted average of the cluster-level mean outcomes, where the weight of each cluster equals the \textit{sum} of individual inverse propensity scores up to a normalizing constant $M_i-1$.
When there is only a single unit within each cluster, i.e., $M_i = 1$ for all $i$, the estimator reduces to the standard IPW estimator under no interference settings.

Crucially, the \addIPW{} estimator leverages the linear additive assumption of the conditional outcome regression by ensuring that the cluster-level weights scale linearly with the individual-level inverse probability weights, which are typically of reasonable magnitude.
In contrast, the IPW estimator given in Equation~\eqref{equ:Vhat:fullIPW} uses the product of individual-level inverse probability weights, which tends to zero as cluster size increases, leading to a large variance.

We next show that when the propensity scores are known, the \addIPW{} estimator is unbiased for the policy value. 
\begin{proposition}[Unbiasedness]\label{prop:unbiased}
    Under Assumptions~\ref{ass:iidcluster},~\ref{ass:DGP}(a),~\ref{ass:additive},~and~\ref{ass:factorCPS}, and for $\forall\;\pi\in\Pi$, $$\E[\widehat{V}^{\addIPW}(\pi)]=V(\pi).$$ 
\end{proposition}
\begin{proof}
    See Appendix~\ref{appendix:proof:unbiasedness}.
\end{proof}
Below, we provide an additional intuition for this result.
First, using the law of iterated expectation and Assumption~\ref{ass:DGP}(a), we rewrite the value function under Assumption~\ref{ass:additive} as follows,
\begin{eqnarray}\label{equ:simp:estmand}
V(\pi) & = & \mathbb{E}\left[\frac{1}{M_i} \sum_{j=1}^{M_i}\left( g^{(0)}_{j}(\bX_i)+\sum_{k=1}^{M_i}g^{(k)}_{j}(\bX_i)\pi(\bX_{ik})\right)\right].
\end{eqnarray}
Thus, we can estimate the value function by substituting the unknown nuisance parameters $\bg_{j}(\bx)=\left(g^{(0)}_{j}(\bx),g^{(1)}_{j}(\bx),\ldots,g^{(m)}_{j}(\bx)\right)^\top$ with their empirical estimates.

Unconfoundedness (Assumption~\ref{ass:DGP}(a)) enables us to rewrite Equation~\eqref{equ:additive} using observable quantities.
We can then view $\bg_{j}(\bX_i)$ given $\bX_i$ as the coefficients of the treatment vector $\tilde{\bA}_i:=(1,A_{i1},\ldots,A_{iM_i})^\top$ in a unit-specific OLS regression of $Y_{ij}$ on $\tilde{\bA}_i$.
In principle, this regression problem cannot be directly solved due to non-identifiability, as there is only one observation for the $M_i+1$ predictors.
However, we can find the population solution $\bg_{j}(\bX_i)$ that minimizes the mean squared error (MSE), leading to the following MSE minimizer

\begin{equation}\label{equ:PI:oracle:estor}
\mathring{\bg}_j(\bX_i)=\mathbb{E}\left[\tilde{\bA}_{i}\tilde{\bA}_{i}^\top \mid \bX_i\right]^{-1} \mathbb{E}\left[\tilde{\bA}_{i}Y_{ij}\mid \bX_i\right].
\end{equation}

Since the matrix $\mathbb{E}\left[\tilde{\bA}_{i}\tilde{\bA}_{i}^\top \mid \bX_i\right]$ is a function of known propensity scores, we can directly compute it.
Under Assumption~\ref{ass:factorCPS}, this matrix is invertible and its inverse is given by,
\begin{equation}\label{equ:inverse:matrix}
\begin{aligned}
    &\mathbb{E}\left[\tilde{\bA}_{i}\tilde{\bA}_{i}^\top \mid \bX_i\right]^{-1}=\\
    & \left[\begin{array}{ccccc}1+\sum_{k=1}^{M_i} \frac{e_k(1\mid \bX_i) }{1-e_k(1\mid \bX_i)} & -\frac{1}{1-e_1(1\mid \bX_i)} & \cdots & \cdots & -\frac{1}{1-e_{M_i}(1\mid \bX_i)} \\ -\frac{1}{1-e_1(1\mid \bX_i)} & \frac{1}{e_1(1\mid \bX_i)\left(1-e_1(1\mid \bX_i)\right)} & 0 & \cdots & 0 \\ \vdots & 0 & \frac{1}{e_2(1\mid \bX_i)\left(1-e_2(1\mid \bX_i)\right)} & 0 & \vdots \\ \vdots & \vdots & 0 & \ddots & 0 \\ -\frac{1}{1-e_{M_i}(1\mid \bX_i)} & 0 & \cdots & 0 & \frac{1}{e_{M_i}(1\mid \bX_i)\left(1-e_{M_i}(1\mid \bX_i)\right)}\end{array}\right].
\end{aligned}
\end{equation}
Given that $\mathbb{E}\left[Y_{ij}\tilde{\bA}_{i}\mid \bX_i\right]$ is unknown, we replace it with the single realized observation  $(Y_{ij},\tilde{\bA}_{i})$ in Equation~\eqref{equ:PI:oracle:estor}, resulting in the following estimator,
\begin{equation}\label{equ:PI:empirical:estor}
\hat{\bg}_j(\bX_i)=\mathbb{E}\left[\tilde{\bA}_{i}\tilde{\bA}_{i}^\top \mid \bX_i\right]^{-1} \tilde{\bA}_{i} Y_{ij}.
\end{equation}

The unbiasedness of this estimator is immediate:
\begin{equation*}
\E\left[\hat{\bg}_j(\bX_i)\mid \bX_i\right]
=\mathbb{E}\left[\tilde{\bA}_{i}\tilde{\bA}_{i}^\top \mid \bX_i\right]^{-1}\mathbb{E}\left[\tilde{\bA}_{i}Y_{ij}\mid \bX_i\right]
=\mathbb{E}\left[\tilde{\bA}_{i}\tilde{\bA}_{i}^\top \mid \bX_i\right]^{-1} \mathbb{E}\left[\tilde{\bA}_{i}\tilde{\bA}_{i}^\top \mid \bX_i\right]\bg_j(\bX_i)
=
\bg_j(\bX_i).
\end{equation*}
Therefore, the linearity of expectation implies the following unbiased estimator of $V(\pi)$:
\begin{align}
    \widehat{V}(\pi)&=\frac{1}{n}\sum_{i=1}^{n} \frac{1}{M_i} \sum_{j=1}^{M_i}\left( \hat{g}^{(0)}_{j}(\bX_i)+\sum_{k=1}^{M_i}\hat{g}^{(k)}_{j}(\bX_i)\pi(\bX_{ik})\right)\nonumber\\
    &=\frac{1}{n}\sum_{i=1}^{n} \frac{1}{M_i} \sum_{j=1}^{M_i}\tilde{\pi}(\bX_i)^\top\mathbb{E}\left[\tilde{\bA}_{i}\tilde{\bA}_{i}^\top \mid \bX_i\right]^{-1} \tilde{\bA}_{i} Y_{ij}, \label{equ:PI:final:estimator}
\end{align}
where $\tilde{\pi}(\bX_i):=(1,\pi(\bX_{i1}),\ldots,\pi(\bX_{iM_i}))^\top$ is a binary treatment assignment vector under a given policy $\pi$.
Finally, substituting Equation~\eqref{equ:inverse:matrix} into Equation~\eqref{equ:PI:final:estimator} yields our estimator $\widehat{V}^{\addIPW}(\pi)$ given in Equation~\eqref{equ:Vhat:addIPW}.

Equation~\eqref{equ:PI:final:estimator} shows that the \addIPW{} estimator can also be written as a weighted average of individual outcomes based on the inverse of {\it individual-level} propensity scores.
Importantly, the proposed estimator utilizes the data from a cluster whose realized treatment assignment does not agree with the policy.
This contrasts with the IPW estimator given in Equation~\eqref{equ:Vhat:fullIPW} that equals a weighted average of cluster-level mean outcome using the inverse of {\it cluster-level} propensity scores, dropping any cluster whose realized treatment assignment does not match with the policy under consideration.
This difference explains why the variance of the \addIPW{} estimator is much smaller than that of the IPW estimator.
As demonstrated in Section~\ref{sec:regretbound}, this efficiency gain in policy evaluation leads to a better performance of policy learning.

The \addIPW{} estimator is derived by considering the unit-specific least squares regression of outcome on a treatment vector.
However, rather than explicitly fitting the outcome model for estimation, it modifies the weights of the IPW estimator such that it is consistent with the semiparametric outcome model.
A similar technique has been used in the previous literature for off-policy evaluation for online recommendation systems \citep{swaminathan2017off}, the estimation of total treatment effect in a design-based single network interference setting \citep{cortez2022exploiting}, and the estimation of average total treatment effect in bipartite network experiments \citep{harshaw2023design}.
In Section~\ref{sec:extension:obs}, we provide an additional justification of the \addIPW{} estimator by establishing its relation to the efficient semiparametric estimator under Assumption~\ref{ass:additive}.

We emphasize that the proposed unbiased estimator in Equation~\eqref{equ:PI:final:estimator} (as well as the generalized estimator proposed below in Section~\ref{sec:higherorder}) remains valid even without the factored propensity score assumption (Assumption~\ref{ass:factorCPS}).
Specifically, the estimator $\widehat{V}(\pi)$ remains valid so long as the experimental design matrix $\mathbb{E}\left[\tilde{\bA}_{i}\tilde{\bA}_{i}^\top \mid \bX_i\right]$ is invertible, although its expression may differ from $\widehat{V}^{\addIPW}(\pi)$ depending on the propensity score structure.
In scenarios where the treatment assignment ensures that every low-order treatment combination for a cluster has nonzero probability, this matrix is likely to be invertible.
Even if the experimental design matrix is not invertible, we can use the Moore-Penrose pseudoinverse in place of the matrix inverse, yielding an unbiased estimator.

\subsection{Semiparametric Model with Interactions}\label{sec:higherorder}

It is possible to extend our semiparametric additive model by including interactions.
Consider the following polynomial additive model,
\begin{equation}\label{equ:polyadditive}
    \E\left[Y_{ij}(\ba_i) \mid \bX_i\right]=\bg_{j}(\bX_i)^\top \phi(\ba_i),
\end{equation}
where we use the following augmented treatment vector that contains up to $\beta$-order interactions between treatments of different units with $\beta < m_{\max}$ where $m_{\max}$ is an upper bound of $M_i$,
\begin{equation} \label{equ:Amapping:poly}
\phi(\ba_i)=\left(1,\{a_{ij}\}_j,\{a_{ij_1}a_{ij_2}\}_{j_1\neq j_2},\ldots,\{a_{ij_1}a_{ij_2}\ldots a_{ij_{\beta}}\}_{j_1\neq \ldots \neq j_\beta}\right)^\top.
\end{equation}
Under this model, $\bg_{j}(\bx)$ represents the unknown heterogeneous effect function set whose size equals the length of $\phi(\ba_i)$.
We can derive an unbiased estimator for $V(\pi)$ as before,
\begin{equation}\label{equ:polyPI:final:estimator}
    \begin{aligned}
    \widehat{V}(\pi)
    &=\frac{1}{n}\sum_{i=1}^{n} \frac{1}{M_i} \sum_{j=1}^{M_i}\phi({\pi}(\bX_i))^\top\mathbb{E}\left[\phi(\bA_i)\phi(\bA_i)^\top \mid \bX_i\right]^{-1} \phi(\bA_i)Y_{ij}.
    \end{aligned}
\end{equation}

The explicit form of this estimator can be obtained by calculating the inverse of $\mathbb{E}\left[\phi(\bA_i)\phi(\bA_i)^\top \mid \bX_i\right]$, which contains up to the $\beta$-order product of individual-level treatment probabilities.
Since directly computing the inverse of this matrix is tedious, one may obtain the weights for individual outcomes by leveraging the linearity of expectation and unbiasedness property of the estimator \citep[see][for a similar technique]{yu2022estimating}.
Appendix~\ref{app:higherorder} provides an explicit expression of the proposed estimator under the general polynomial additive model.

In principle, $\phi(\bA_i)$ can be extended to a vector of at most length $\sum_{k=0}^{M_i} {M_i \choose k} =2^{M_i}$, which allows for all possible treatment interactions within a cluster of size $M_i$.
Under this extreme scenario, which implies no assumption about spillover effects, the proposed estimator can be shown to be equal to the IPW estimator given in Equation~\eqref{equ:Vhat:fullIPW}.
In practice, however, researchers must choose the value of polynomial order $\beta$ by considering a bias-variance tradeoff \citep[see also][who proposes a similar low-order interaction model in a design-based single network context]{cortez2022exploiting}.
We can further extend our model by letting $\beta$ depend on cluster size $M_i$. This approach will allow for the inclusion of fewer interactions in smaller clusters.

Our experience suggests that in most cases the linear or quadratic additivity assumption is sufficient for effective policy evaluation and learning.
Formally, when the true model includes higher-order interactions, our estimator based on the linear additive assumption can be interpreted as the following approximation to the true policy value $V(\pi)$ given in Equation~\eqref{equ:V}, i.e.,
\begin{equation}\label{equ:projection}
    \E\left[\frac{1}{M_i} \sum_{j=1}^{M_i}\bg_{j}^{\text{Proj.}}(\bX_i)^\top\tilde{\pi}(\bX_i)\right],
\end{equation}
where $\bg_{j}^{\text{Proj.}}(\bx)$ is the projection of each unit's true outcome function $\mu_j(\bA_i, \bX_i)=\E[Y_{ij}\mid \bA_i, \bX_i]$ onto the linear treatment vector space,
\begin{equation*}
    \bg_{j}^{\text{Proj.}}(\bX_i)=\arginf_{\bg}\;\mathbb{E}\left[\left(\mu_j(\bA_i,\bX_i)-\bg(\bX_i)^\top\tilde{\bA}_i\right)^2 \mid \bX_i\right].
\end{equation*}

For simplicity, we assume the linear additivity semiparametric model (i.e., Assumption~\ref{ass:additive}) throughout this paper and leave the data-driven choice of $\beta$ to future work.


\section{Policy Learning}\label{sec:empiricalPL}

We now consider the problem of policy learning.
Specifically, we solve the following empirical analog of the optimization 
problem given in Equation~\eqref{equ:piast} using our \addIPW{} estimator in Equation~\eqref{equ:Vhat:addIPW},
\begin{equation}\label{equ:pihat}
    \begin{aligned}
      \pihat:=  \underset{\pi\in\Pi}{ \operatorname{argmax}}\  \widehat{V}^{\addIPW}(\pi).
    \end{aligned}
\end{equation}
We first measure the learning performance of $\pihat$ by deriving a non-asymptotic upper bound on the true population regret of $\pihat$ defined in Equation~\eqref{equ:regret:def}, i.e., $R(\pihat)$.
We then show that the optimization problem can be solved using a mixed-integer linear program formulation.

\subsection{Regret Analysis}\label{sec:regretbound}

We establish a finite-sample regret bound for $\pihat$, assuming that the propensity scores are known.
In Section~\ref{sec:extension:obs}, we extend our theoretical results to observational studies where propensity scores are unknown and must be estimated.
We begin by stating the following standard assumptions.
\begin{assumption}\label{ass:regret:bound}
The following statements hold:
    \begin{enumerate}[label=(\alph*)]
\item \textit{Bounded outcome}: $\exists$ a constant $B\geq0$ such that $\lvert Y_{ij}(\ba_i)\rvert\leq B$ for all $\ba_i \in \mathcal{A}\left(M_i\right)$
\item \textit{Finite cluster size}: $\exists$ $m_{\max } \in \mathbb{N}$ such that $M_i \leq m_{\max }$ almost surely
\item \textit{Bounded complexity of policy class}: consider a policy class $\Pi$ of binary-valued functions $\pi:\mathcal{X}\rightarrow\{0,1\}$ that has a finite VC dimension $\nu<\infty$ 
\end{enumerate}
\end{assumption}

Assumption~\ref{ass:regret:bound}(a) is standard in the literature.
Assumption~\ref{ass:regret:bound}(b) restricts cluster size $M_i$ to be bounded, implying that cluster size is not too large relative to the number of clusters $n$.
The proposed methodology may not perform well when the number of clusters is small.
Assumption~\ref{ass:regret:bound}(c) restricts the complexity of the policy class $\Pi$ of interest using the concept of VC dimension \citep{vapnik2015uniform}.
This assumption is often made in the existing policy learning literature to avoid overfitting \citep[e.g.,][]{kitagawa2018should,athey2021policy}.
The assumption holds for common policy classes such as linear and fixed-depth decision trees.

\begin{theorem}[Finite-sample regret bound]\label{thm:regret}
Suppose Assumptions~\ref{ass:iidcluster},~\ref{ass:DGP}(a),~and~\ref{ass:additive}--\ref{ass:regret:bound} hold. Define $\pihat$ as the solution to Equation~\eqref{equ:pihat}.
For any $\delta>0$, with probability at least $1-\delta$, the regret of $\pihat$ can be upper bounded as,
\begin{equation}
    R(\pihat)\leq \frac{4C}{\sqrt{n}}+4c_0\frac{Bm_{\max}}{\eta}\sqrt{\frac{\nu}{n}}+2C\sqrt{\frac{2}{n}\log\frac{1}{\delta}}
\end{equation}
where $C=B[m_{\max}\times(\frac{1}{\eta}-1)+1]$ and $c_0$ is a universal constant.
\end{theorem}
\begin{proof}
    See Appendix~\ref{appendix:proof:finiteregretbound}.
\end{proof}

Theorem~\ref{thm:regret} provides a finite-sample upper bound on the regret, under the case of known propensity scores. 
The regret converges to zero at the rate of $1/\sqrt{n}$, which matches the optimal regret rate for i.i.d. policy learning \citep{kitagawa2018should,athey2021policy}.
Moreover, the bound linearly depends on the maximal cluster size, and is inversely proportional to the lower bound of the individual-level propensity score $\eta$.
This result is distinct from the regret bound based on the standard IPW estimator given in Equation~\eqref{equ:Vhat:fullIPW}, which is typically of order $O_p\left(\frac{B}{\eta^{m_{\max}}} \sqrt{\frac{\nu}{n}}\right)$.
While Theorem~\ref{thm:regret} assumes that the outcome model satisfies Assumption~\ref{ass:additive}, we find that the learned optimal policy \( \pihat \) in Equation~\eqref{equ:pihat} often achieves satisfactory learning performance in practice, even when the outcome model is misspecified (see Section~\ref{sec:simulation}). In such a misspecified case, \( \pihat \) optimizes the best linear semiparametric approximation of the true value function, given in Equation~\eqref{equ:projection}, reflecting an important bias-variance tradeoff in policy learning.

\subsection{Mixed-Integer Program Formulation}\label{sec:MIP}

The above results hold for the exact solution to Equation~\eqref{equ:pihat}. In general, solving Equation~\eqref{equ:pihat} leads to a nonconvex optimization problem, which is difficult to solve in practice.
For a certain policy class, however, we can considerably simplify the optimization problem using a mixed-integer program (MIP) formulation.
Such policy classes include linear decision rules, fixed-depth decision trees, and treatment sets with piecewise linear boundaries \citep[e.g.,][]{kitagawa2018should,viviano2019policy,athey2021policy,zhou2023offline}.

For example, consider a linear policy rule of the following form:
\[
\Pi=\{\pi:\mathcal{X}\rightarrow \{0,1\}\mid \pi(\bX) = \mathds{1}(\bX^\top\bbeta\geq0),\quad \bbeta\in\mathcal{B}\}.
\]
Following \cite{kitagawa2018should}, we introduce binary variable $p_{ij}$ and write 
\[
\frac{\bX_{ij}^\top \bbeta}{C_{ij}}<   p_{ij} \leq 1+\frac{\bX_{ij}^\top\bbeta}{C_{ij}}, \quad\quad C_{ij}>\sup _{\bbeta\in \mathcal{B}}\left|\bX_{ij}^\top \bbeta\right|, \quad\quad p_{ij} \in\{0,1\}.
\]
Then, $p_{ij}$ is equal to one if $\bX_{ij}^\top\bbeta$ is non-negative and zero otherwise, i.e., $p_{ij}=\pi(\bX_{ij})$. 
We now can write the objective function (up to some constants) as,
\[
 \frac{1}{n}\sum_{i=1}^{n}\overline{\bY}_{i}\sum_{j=1}^{M_i}
  \left( \frac{A_{ij}}{  e_j(1 \mid\bX_{i})}-\frac{1-A_{ij}}{  e_j(0 \mid\bX_{i})}\right)p_{ij}.
\]
This implies that Equation~\eqref{equ:pihat} can be equivalently represented as the following mixed-integer linear program (MILP), which can be solved using an off-the-shelf algorithm:
\begin{equation}
    \begin{aligned}
&\max_{\substack{\bbeta \in \mathcal{B}, \{p_{ij}\} \in \mathbb{R}}}  \; \frac{1}{n}\sum_{i=1}^{n}\overline{\bY}_{i}\sum_{j=1}^{M_i}
  \left( \frac{A_{ij}}{  e_j(1 \mid\bX_{i})}-\frac{1-A_{ij}}{  e_j(0 \mid\bX_{i})}\right)p_{ij}.\\
&\text { s.t. } 
\begin{aligned}
   & \frac{\bX_{ij}^\top \bbeta}{C_{ij}}<   p_{ij} \leq 1+\frac{\bX_{ij}^\top\bbeta}{C_{ij}} \quad \text { for } i=1, \ldots, n, \quad\text{and }\quad j=1,\ldots,M_i,\\
& p_{ij} \in\{0,1\},
\end{aligned}
\end{aligned}
\end{equation}
where constants $C_{ij}$ should satisfy $C_{ij}>\sup _{\bbeta \in \mathcal{B}}\left|\bX_{ij}^\top \bbeta\right|$. 

While this MIP formulation enables exact optimization for many policy classes, solving large-scale MIP problems can be computationally demanding, especially in settings with many clusters or large cluster sizes. As a computationally scalable alternative, we also consider a smooth stochastic approximation approach. Following \citet{fang2022fairness}, we approximate the binary deterministic ITR $\pi(\bX_{ij})$ with a logistic function $f(\bX_{ij}) = {1 + \exp(-\bX_{ij}^\top \bbeta)}^{-1}$. This replaces the discrete decision rule with a continuous and differentiable surrogate, enabling the use of fast gradient-based optimization algorithms.

\section{Extension to Observational Studies}\label{sec:extension:obs}

So far, we have focused on the experimental setting in which the propensity score is known.
In this section, we extend our methodology to observational studies in which the propensity score is unknown and must be estimated.
Consider a plug-in approach that directly replaces the true propensity score $e_j$ with its estimate $\hat{e}_j$ in the policy value estimator given in Equation~\eqref{equ:Vhat:addIPW}.
It can be shown that the resulting regret of $\pihat$ depends on the estimation error of unknown propensity scores (see Theorem~\ref{thm:regret:plugin} in Appendix~\ref{appendix:unknownPS}).
This implies that in all but the simplest cases, the plug-in approach will result in a sub-optimal rate, which is slower than $1 / \sqrt{n}$, for the learned policy $\hat{\pi}$. 

Given this suboptimality of the plug-in approach, we develop alternative efficient policy evaluation and learning methods by building on \cite{chernozhukov2019semi}, who studies efficient policy learning with continuous actions.
Specifically, we specialize their approach and propose a doubly robust policy value estimator under our linear additive structural assumption.
As shown below, this approach is robust to estimation errors of the propensity score model or the outcome regression model. Moreover, it attains the semiparametric efficiency bound under Assumption~\ref{ass:additive}, provided the estimation errors of the nuisance functions satisfy mild rate conditions.
Finally, we show that the established asymptotic regret bound of the learned optimal policy based on the doubly robust estimator achieves a fast $1 / \sqrt{n}$ convergence rate.

\subsection{Doubly Robust Estimator}

To define the doubly robust estimator, let $\mu(\ba,\bx)= \E\left[\bY_{i}\mid \bA_i=\ba,\bX_i=\bx\right]$ be the vector of true conditional expected outcomes in cluster $i$.
Using this notation, we rewrite Assumption~\ref{ass:additive} as,
\begin{equation}\label{equ:semi:model}
 \mu(\ba,\bx) =G(\bx) \phi(\ba)
\end{equation}
where $\phi(\ba):=(1,a_1,\ldots,a_m)^\top$ is the treatment assignment vector, and $G(\bx)$ is a $m\times(m+1)$ unknown function matrix of the following form,
\[
G(\bx):=\left(\begin{array}{c} \bg_{1}^\top(\bx) \\ \vdots \\\bg_{m}^\top(\bx) \end{array}\right)=\left(\begin{array}{ccc} g_1^{(0)}(\bx) & \ldots & g_1^{(m)}(\bx)\\ \vdots &&\vdots \\  g_m^{(0)}(\bx) & \ldots & g_m^{(m)}(\bx) \end{array}\right).
\]
Following the strategy described in Section~\ref{sec:higherorder}, it is possible to increase the model complexity by augmenting the treatment vector with interaction terms (see Equation~\eqref{equ:polyadditive}).
For the sake of simplicity, however, we assume the linear additive model in this section.

We also define the conditional covariance matrix of $\phi(\ba)$ as $\Sigma(\bx)=\E[\phi(\ba)\phi(\ba)^\top\mid \bx]$. 
The doubly robust estimator we propose below, will rely on estimates of these two nuisance functions.
Notice that $\Sigma(\bX_i)$ is exactly equal to the matrix $\mathbb{E}\left[\tilde{\bA}_{i}\tilde{\bA}_{i}^\top \mid \bX_i\right]$ defined in Section~\ref{sec:addIPW}, except that it now involves unknown propensity scores.
Due to the factorized propensity score assumption (Assumption~\ref{ass:factorCPS}), $\Sigma$ can be inverted, and the resulting expression is given by Equation~\eqref{equ:inverse:matrix}. 

Based on Equation~\eqref{equ:semi:model}, the policy value is identified as $V(\pi)=\E\left[ w(M_i)^\top G(\bX_i) \phi(\pi(\bX_i)) \right]$,
where $w(M_i)=\frac{1}{M_i}\mathbf{1}_{M_i}$ are the uniform weights for averaging the outcomes of units within the cluster and $\pi(\bX_i)=(\pi(\bX_{i1}),\ldots,\pi(\bX_{iM_i}))^\top$ is the vector of treatment assignment under policy $\pi$.
We propose the following doubly-robust estimator (\addDR),
\begin{equation}\label{equ:DRestor}
    \begin{aligned}
           \widehat{V}^{\addDR}(\pi) 
          &=\frac{1}{n}\sum_{i=1}^{n} w(M_i)^\top \widehat{G}_\addDR(\bY_i,\bX_i)\phi(\pi(\bX_i)),
    \end{aligned}
\end{equation}
where
\begin{equation}\label{equ:DRscore}
    \begin{aligned}
\widehat{G}_\addDR(\bY_i,\bX_i) = \widehat{G}(\bX_i)+\left(\bY_i-\widehat{G}(\bX_i)\phi(\bA_i)\right)\phi(\bA_i)^\top\widehat{\Sigma}(\bX_i)^{-1} 
      \end{aligned}
\end{equation}  
can be viewed as an estimate of $G(\bX_i)$ based on a single observation, and $ \widehat{G}$ and $\widehat{\Sigma}$ are the estimates for the nuisance quantities $G$ and $\Sigma$.
Equation~\eqref{equ:DRscore} has a form similar to the standard doubly robust estimators in the literature, which typically consist of an outcome regression estimate plus an augmented weighted residual term.
It can be easily shown that $\widehat{V}^{\addDR}(\pi)$ enjoys a doubly robust property --- it is consistent if either of the two nuisance models is consistently estimated. Therefore, the \addDR{} estimator is more robust to the estimation errors of propensity score and outcome regression models.
In addition, if we substitute $\widehat{G}_\addDR(\bY_i,\bX_i)$ in Equation~\eqref{equ:DRestor} with ${G}_\addIPW(\bY_i,\bX_i)=\bY_i\phi(\bA_i)^\top{\Sigma}(\bX_i)^{-1} $, this policy value estimator reduces to our \addIPW{} estimator, providing an additional justification of the proposed estimator under the experimental settings.

\subsection{Estimation of Nuisance Functions}


We give an example of constructing the \addDR{} estimator in practice.
Under Assumption~\ref{ass:factorCPS}, the formula of $\widehat{V}^{\addDR}(\pi)$ simplifies to
\begin{equation}\label{equ:DRestor:closedform}
    \begin{aligned}
        \widehat{V}^{\addDR}(\pi) &= \frac{1}{n}\sum_{i=1}^{n} \ \hat{\overline{\bg}}^\top(\bX_i)\phi(\pi(\bX_i)) + \left(\overline{\bY}_{i}- \hat{\overline{\bg}}^\top(\bX_i)\phi(\bA_i)\right) \left[\sum_{j=1}^{M_i}
  \left( \frac{\mathds{1}\left\{
       A_{ij}=\pi(\bX_{ij})
     \right\}}{  \hat{e}_j(
 \pi(\bX_{ij})
      \mid \bX_{i})}
  -1\right)+1\right]
    \end{aligned}
\end{equation}
where $\overline{\bg}^\top(\bx):=w(m)^\top G(\bx)=\frac{1}{m}\sum_{j=1}^{m}\bg^\top_{j}(\bx)$ denotes the vector of unknown function sets in the cluster-level average outcome model.
Thus, in this case, the construction of the \addDR{} estimator reduces to obtaining the nuisance estimates $\hat{\overline{\bg}}(\cdot)$ and $\{\hat{e}_j(\cdot)\}_{j=1}^m$.

We mainly discuss the estimation of $\overline{\bg}(\bx):=\left(\overline{g}^{(0)}(\bx),\overline{g}^{(1)}(\bx),\ldots,\overline{g}^{(m)}(\bx)\right)$ and $\{e_j(\bx)\}_{j=1}^m$ for $\bx\in\mathcal{X}$ using nonparametric sieve estimators \citep{newey1997convergence,chen2007large}.
Since cluster sizes $m$ vary and clusters with significantly different sizes may exhibit different behaviors, we propose stratifying the data based on cluster size and fitting the nuisance models separately for each $m$.
In practice, when sample size is limited for some clusters, one can alternatively fit a universal model using cluster size as one of the covariates in the model.

We employ cross-fitting \citep{chernozhukov2018double}. 
Specifically, we randomly split the sample of clusters into several disjoint folds such that each fold contains clusters of all cluster sizes, and the proportion of each cluster size type is nearly identical across different folds (see also Section~4.2 in \cite{park2022efficient}).
For each fold, we train the nuisance models on the remaining folds and predict the nuisance estimates on the held-out fold. This process is repeated for every fold, yielding nuisance predictions on the entire dataset.
For simplicity, we omit explicit dependence on the folds in our notation and focus on the estimation strategy throughout this section.

We define the conditional cluster-level average outcome model as $\overline{\mu}(\ba,\bx)= \E[\overline{\bY}_{i}\mid \bA_i=\ba,\bX_i=\bx]=\overline{\bg}(\bx)^\top\phi(\ba)=\overline{g}^{(0)}(\bx)+\sum_{j=1}^{m}\overline{g}^{(j)}(\bx)a_j$. Our goal is to estimate the nonparametric functions $\left(\overline{g}^{(0)}(\bx),\overline{g}^{(1)}(\bx),\ldots,\overline{g}^{(m)}(\bx)\right)$ separately for each stratum defined by the cluster size $m$.
Let $\{r_{m,k}(\bx)\}_{k=1}^\infty$ be a sequence of known basis functions (e.g, polynomials, splines). 
We impose a structural assumption on $r_{m,k}(\bx_j,\bx_{(-j)})$, requiring it to be permutation invariant with respect to the covariates of the remaining $(m-1)$ units in the same cluster. For example, $r_{m,k}(\bx_j,\bx_{(-j)})$ can be defined as a function of $\bx_j$ and summary statistics of $\bx_{(-j)}$ within each subset $\{(-j)\}$, such as the mean or second moment.
Let $\hat{\overline{g}}_K^{(j)}(\bx)$ denote
the estimator for $\overline{g}^{(j)}(\bx)$, using the first $K$ basis functions $R_{m,K}(\bx_j,\bx_{(-j)}):=\left(r_{m,1}(\bx_j,\bx_{(-j)}),\ldots,r_{m,K}(\bx_j,\bx_{(-j)}) \right)^\top$, given in the form of
\[
\hat{\overline{g}}_K^{(j)}(\bx)=R_{m,K}(\bx_j,\bx_{(-j)})^\top \hat{\btheta}_{m,K},\quad j = 1, \ldots, m,
\]
where $\hat{\btheta}_{m,K}$ represents the coefficients to be estimated.
We assume that the coefficient parameter \( \btheta_{m,K} \) does not depend on the unit index $j$.
This assumption is reasonable when the heterogeneity of spillover effects can be fully captured by the covariates \( \{\bx_j\}_{j=1}^m \) of the units. 
Intuitively, $\hat{\overline{g}}_K^{(j)}(\bx)$ provides a better approximation of ${\overline{g}}^{(j)}(\bx)$ as $K$ increases. 
Similarly, the intercept term $\overline{g}^{(0)}(\bx)$ is approximated using a sieve estimator $\hat{\overline{g}}_K^{(0)}(\bx)=R_{m,K}(\bx)^\top \hat{\bgamma}_{m,K}$. 

Finally, we estimate ${\btheta}_{m,K}$ and ${\bgamma}_{m,K}$ by fitting an OLS regression to clusters of the same size:
\begin{equation}
          \overline{\mu}(\bA_{i}, \bX_i) \ = \   R_{m,K}(\bX_i)^\top\cdot {\bgamma}_{m,K} + \sum_{j=1}^{m} A_{ij} \cdot R_{m,K}(\bX_{ij},\bX_{i(-j)})^\top {\btheta}_{m,K},
\end{equation}
where the summation term is permutation invariant.
This approach eliminates the need to enumerate individual units, as the basis functions are permutation invariant with respect to the covariates of neighboring units, and the effect coefficients $\btheta_{m,K}$ and $\bgamma_{m,K}$ do not depend on $j$. 
        
Similarly, we estimate the propensity scores $\{{e}_j(\bx)\}_{j=1}^m$ using (a possibly different) $K$ approximation functions $R_{m,K}(\bx)$.
To simplify the estimation process and avoid the need to enumerate units within a cluster, we may assume a universal model for all units $j$ for clusters of size $m$:
\begin{equation}
       {e}_j(\bx) = \Prob\left(A_{ij}=1 \mid \bX_{ij}=\bx_{j},\bX_{i(-j)}=\bx_{(-j)}\right) = \operatorname{expit}\left(R_{m,K}(\bx_j,\bx_{(-j)})^\top\btau_{m,K}\right)
\end{equation}
where \( \operatorname{expit}(z) = \frac{1}{1 + e^{-z}} \) is the sigmoid function. The parameter \( \btau_{m,K} \) is then estimated by fitting a logistic regression across all units within clusters of size $m$.

In the above example, we focus on nonparametric series estimators, which are consistent under regularity conditions. 
Notably, \cite{qu2021efficient} also employs nonparametric series estimators for nuisance models.  However, their approach assumes the conditional exchangeability of potential outcomes, which is distinct from our semiparametric structural assumption.
In practice, alternative parametric or nonparametric estimators, such as matching, kernel regression, and other machine learning methods, may also be used to estimate the propensity and outcome models. 


\subsection{Theoretical Analysis}

Next, we establish the theoretical properties of the doubly robust estimator while remaining agnostic to the specific method used to obtain nuisance estimates $\hat{\overline{\bg}}(\cdot)$ and $\widehat{\Sigma}(\cdot)$ and simply imposing high-level conditions on their convergence rates.
Throughout, we present the theoretical results for the simpler case where the nuisance estimates are trained on an independent, separate data split. However, these results qualitatively extend to the case where the cross-fitting technique is applied.

\begin{assumption}[Convergence rates of estimated nuisance functions]\label{ass:rate}
   For any fixed cluster size $M_i=m\in\{1,\ldots,m_{\max}\}$, the nuisance estimates $\hat{\overline{\bg}}(\cdot)$ and $\widehat{\Sigma}(\cdot)$ satisfy 
   $$
   \sup_{\bx}\left\{\left|\hat{\overline{g}}^{(j)}(\bx)-{\overline{g}}^{(j)}(\bx)\right|\right\} \xrightarrow{p} 0\quad \text{for}\ j=0,\ldots,m, \quad  \sup_{\bx} \left\|\widehat{\Sigma}(\bx)-\Sigma(\bx)\right\|_{\infty} \xrightarrow{p} 0
   $$
   and 
    $$
\begin{aligned}
 \E\left[\left(\hat{\overline{g}}^{(j)}(\bX)-{\overline{g}}^{(j)}(\bX)\right)^2
 \right] & = O_{p}\left(r_{n,g}^2\right)\quad \text{for}\ j=0,\ldots,m, \\
\E\left[ \left\|\widehat{\Sigma}(\bX)-\Sigma(\bX)\right\|^2_{F}\right]& =O_{p}\left(r_{n,\Sigma}^2\right),
\end{aligned}
$$
where $r_{n,g}=o(1)$, $r_{n,\Sigma}=o(1)$ and $r_{n,g}\cdot r_{n,\Sigma}=o(1/\sqrt{n})$. Here, $\|\cdot\|_{F}$ represents the Frobenius norm of a matrix.
\end{assumption}

Assumption~\ref{ass:rate} requires that the nuisance estimators $\hat{\overline{\bg}}(\cdot)$ and $\widehat{\Sigma}(\cdot)$ are uniformly consistent and that the product of their mean squared errors (MSE) achieves the $o_p(1/\sqrt{n})$ rate.
Such conditions have been extensively used in the semiparametric estimation literature \citep[e.g.,][]{newey1997convergence,chernozhukov2018double} and can be satisfied by many machine learning and nonparametric estimators, depending on the regularity of the underlying function being estimated.
In experimental studies where the true propensity scores are known, the estimated outcome model can converge to the true outcome model at any rate while still satisfying Assumption~\ref{ass:rate}.

Under Assumptions~\ref{ass:additive}~and~\ref{ass:rate} and an additional assumption of homoskedastic error, the variance of the doubly robust estimator achieves the semiparametric efficiency bound for estimating the policy value of a given policy.
\begin{assumption}[Homoskedasticity]\label{ass:homo}
Assume homoscedastic error in the outcome model, i.e.,
\[\mathbb{E}\left[ (\overline{\bY}- \overline{\mu}(\bA,\bX) )(\overline{\bY}- \overline{\mu}(\bA,\bX) )\mid \bA,\bX\right]=\mathbb{E}\left[ (\overline{\bY}- \overline{\mu}(\bA,\bX) )(\overline{\bY}- \overline{\mu}(\bA,\bX) )\right].\]
\end{assumption}
The assumption of homoskedastic error in the residual function is often seen in the semiparametric literature \citep{robinson1988root,ai2003efficient,chamberlain1992efficiency}.
\begin{theorem}[Semiparametric Efficiency]\label{thm:semieff}
Under Assumptions~\ref{ass:iidcluster}-\ref{ass:additive},~\ref{ass:regret:bound}-\ref{ass:homo}, $\widehat{V}^{\addDR}(\pi)$ is semiparametrically efficient for estimating $V(\pi)$ for any given $\pi$.
\end{theorem}
\begin{proof}
    See Appendix~\ref{appendix:pf:semipara}.
\end{proof}
This result suggests that the \addDR{} estimator can yield efficient policy value estimates in observational studies.
In experimental studies where the propensity score is known, our estimator \addDR{} can achieve a greater reduction in variance in the estimation of the policy value than \addIPW{} estimator while maintaining unbiasedness.

Finally, we perform policy learning based on this \addDR{} estimator: 
\[
\pihat^{\addDR}:=  \underset{\pi\in\Pi}{ \operatorname{argmax}}\  \widehat{V}^{\addDR}(\pi).\]
In Appendix~\ref{appendix:MILP:DR}, we show that this policy optimization problem can also be formulated as a linear MIP problem, enabling the use of off-the-shelf algorithms.
More importantly, we establish the asymptotic regret bound for $\pihat^{\addDR}$, demonstrating its fast convergence rate with optimal dependence on both the sample size $n$ and the policy class $\Pi$.


\begin{theorem}[Regret bound with doubly robust estimator]\label{thm:regret:DR}
    Suppose Assumptions~\ref{ass:iidcluster}-\ref{ass:additive},~\ref{ass:regret:bound}-\ref{ass:rate} hold. The regret of $\pihat^{\addDR}$ can be upper bounded as,
\begin{equation*}
    R(\pihat^{\addDR})\leq O_p\left(\sqrt{\frac{\nu}{n}}\right).
\end{equation*}
\end{theorem}
\begin{proof}
    See Appendix~\ref{appendix:pf:DRregret}.
\end{proof}
The regret bound is asymptotic in the number of clusters $n$ as in the literature \citep{athey2021policy,zhou2023offline}. Furthermore, it achieves regret guarantees with optimal dependence on both the sample size $n$ and the policy class complexity $\nu$, matching the fast convergence rate established in Theorem~\ref{thm:regret} under the assumption of known propensity scores.

\section{Simulation Studies}\label{sec:simulation}

We conduct simulation studies to assess the finite-sample learning performance of our proposed methodology.  We compare this against the performance of the oracle policy, the learned policy under anonymous interference, and the one based on the standard IPW estimator.
We also examine the finite-sample performance of the proposed policy value estimators for policy evaluation.

\subsection{Setup}
We generate $n\in \{50,100,200,400,800\}$ clusters, and for each cluster $i$, we randomly generate cluster size $M_i\in\{5,10,15\}$ with uniform probability. 
For each unit $j$, we independently sample four covariates  $(X_{i j 1},\ldots,X_{i j 4})$ from the standard normal distribution. We generate the treatment variable $A_{ij}$ from independent Bernoulli distributions with success probability of 0.3.
Thus, our propensity score model satisfies Assumption~\ref{ass:factorCPS}.
Throughout, we assume that the propensity score is known.
Lastly, we sample the outcome variable $Y_{ij}$ from the following two models.
\begin{equation}
\begin{aligned}\label{equ:sim:DGP}
  (A)\quad   & Y_{i j}\mid \bA_i, \bX_i \sim N\left(\mu_{i j},1\right), \quad \text{where}\\
  & 
      \mu_{ij}=
(X_{ij1}+0.5X_{ij2}-X_{ij3}-0.5X_{ij4}) A_{ij}
+ 1.5\frac{ \sum_{j^\prime\neq j} (X_{ij'3} +X_{ij'4})A_{ij^\prime} }{M_i-1}+0.2X_{ij2}+0.2X_{ij3},\\ 
(B)\quad   & Y_{i j}\mid \bA_i, \bX_i \sim N\left(\mu_{i j},1\right), \quad \text{where}\\
 &
 \mu_{ij}=
(X_{ij1}+0.5X_{ij2}-X_{ij3}-0.5X_{ij4}) A_{ij}
+ 1.5\frac{ \sum_{j^\prime\neq j} (X_{ij'3} +X_{ij'4})A_{ij^\prime} }{M_i-1} +0.2X_{ij2}+0.2X_{ij3}\\ 
& \qquad -1.5(X_{ij1}^2+ X_{ij2}^2)A_{ij}\overline{A}_{i(-j)},
 \end{aligned}
\end{equation}
The outcome regression model under Scenario~A above satisfies Assumption~\ref{ass:additive}, while the outcome model under Scenario~B, which includes interaction terms, does not.

We consider the following class of linear thresholding policies,
\[
\pi(\bX) = \mathds{1}(\beta_0+\beta_1X_1+\beta_2X_2+\beta_3X_3+\beta_4X_4\geq0).
\]
Within this policy class, we find the best policy using our proposed \addIPW{} and \addDR{} policy evaluation estimators.
Since the propensity scores are known, we only need to train the cluster-level outcome regression model for the \addDR{} estimator. We fit a linear regression model using all cluster-level covariates and their interaction terms with cluster-level treatments.

For comparison, we also consider optimal policies learned based on three standard IPW estimators.
The first is the IPW estimator with {\it unknown interference} $\widehat{V}^{\textsf{IPW}}(\pi)$, which is given in Equation~\eqref{equ:Vhat:fullIPW} and makes no assumption about the interference structure.
The second is the IPW estimator with the widely used {\it anonymous interference} assumption, denoted by $\widehat{V}^{\textsf{AnonyInt}}(\pi)$, which assumes that one's outcome depends only on the proportion of the treated units in the same cluster,
$$\widehat{V}^{\textsf{AnonyInt}}(\pi)= \frac{1}{n}\sum_{i=1}^{n} \frac{1}{M_i}\sum_{j=1}^{M_i}
 \frac{\mathds{1}\left\{
       A_{ij}=\pi(\bX_{ij}), \sum_{k\neq j } A_{ik}=\sum_{k\neq j } \pi(\bX_{ik})
     \right\}}{  \Prob\left(
   A_{ij}=\pi(\bX_{ij}), \sum_{k\neq j } A_{ik}=\sum_{k\neq j } \pi(\bX_{ik})
      \mid \bX_{i}\right)}Y_{ij}.$$
The third is the standard IPW estimator with {\it no interference} \citep{kitagawa2018should}, which 
is slightly adjusted for our cluster setting, i.e., 
\begin{equation*}
\widehat{V}^\textsf{NoInt}(\pi)= \frac{1}{n}\sum_{i=1}^{n} \frac{1}{M_i}\sum_{j=1}^{M_i}
 \frac{\mathds{1}\left\{
       A_{ij}=\pi(\bX_{ij})
     \right\}}{  e_j(
 \pi(\bX_{ij})
      \mid \bX_{i})}Y_{ij}.
      \end{equation*}

For our proposed estimators and the standard IPW estimator with no interference $\widehat{V}^\textsf{NoInt}(\pi)$, we formulate the optimization problem as an MILP problem, which we efficiently solve using a stochastic approximation approach described in Section~\ref{sec:MIP}. 

In contrast, the estimators, $\widehat{V}^{\textsf{IPW}}(\pi)$ and $\widehat{V}^{\textsf{AnonyInt}}(\pi)$, lead to more challenging optimization problems.  The former involves the cluster-level policy indicator variable, i.e., $\mathds{1}\left\{  \bA_i=\{\pi(\bX_{ij})\}_{j=1}^{M_i}\right\}$, while the latter involves the indicator of the number of treated neighbors within the same cluster, i.e., $\mathds{1}\left\{
      \sum_{k\neq j } A_{ik}=\sum_{k\neq j } \pi(\bX_{ik})
     \right\}$.
To address these challenges, we again apply stochastic approximation techniques. For the cluster-level policy indicator, we use the stochastic ITR approximation described above. For the anonymous interference estimator $\widehat{V}^{\textsf{AnonyInt}}(\pi)$, we smooth the treated neighbors' indicator using an exponential weighting function, where each possible number of treated neighbors $h$ is assigned a weight of $\exp \left\{-\left(h-\sum_{k\neq j }\pi(\bX_{ik})\right)^2\right\}$.

We also include the oracle estimator that directly optimizes the empirical policy value based on the true outcome model given in Equation~\eqref{equ:sim:DGP}, using the same stochastic approximation approach.
Finally, all approximate optimization problems are solved using a gradient-based L-BFGS algorithm.

For each simulation setup, we generate 4,000 independent data sets and examine the performance of the aforementioned six estimators: the proposed \addIPW{} and \addDR{} estimators, the IPW estimators with unknown, anonymous, and no interference, and the oracle estimator.
Since the true policy value under each learned policy cannot be easily calculated analytically, we separately obtain an additional sample of 10,000 clusters from the assumed outcome model and approximate the true policy value. 

\subsection{Results}

\begin{figure}[!t]
    \centering
    \includegraphics[width=1\linewidth]{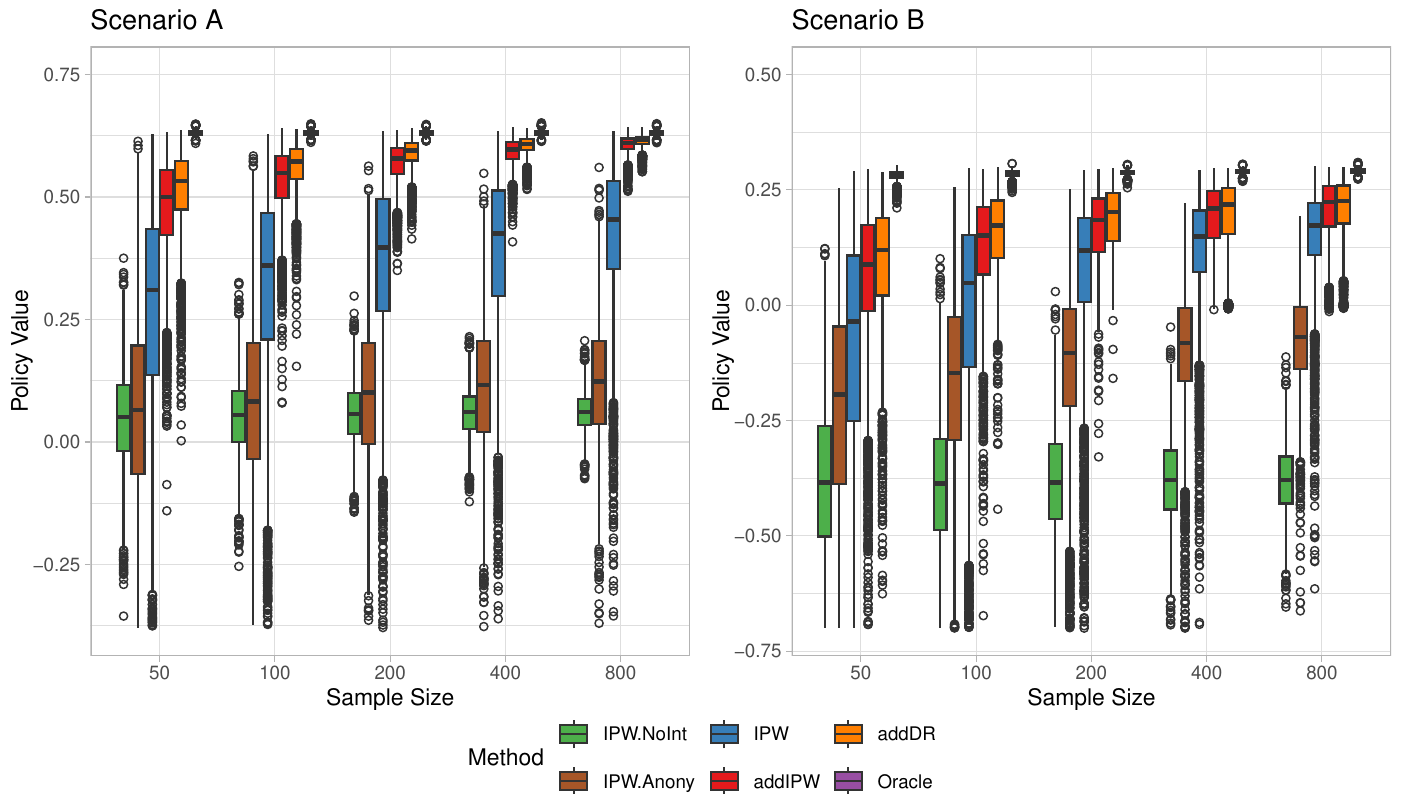}
    \caption{Boxplots of the policy value under the learned policies based on the proposed \addIPW{} (red) and \addDR{} (orange) estimators, in comparison to the standard IPW estimators with unknown interference (blue), anonymous interference (brown), and no interference (green). The value of the oracle estimator (purple) is also shown. Even when the model is misspecified (Scenario~B), the proposed policy learning methods outperform the IPW estimators, with its estimated policy values closest to those of the oracle estimator. The \addDR{} estimator further improves upon the \addIPW{} estimator by reducing variance.}
 \label{fig:sim:PL}
\end{figure}

The left panel of Figure~\ref{fig:sim:PL} shows that under Scenario~A where Assumption~\ref{ass:additive} is met, the proposed methodology, based on the \addIPW{} (red boxplot) and \addDR{} (orange) estimators, outperforms policies learned based on the standard IPW estimators.
Since the \addIPW{} and \addDR{} estimators are both unbiased in this setting, the performance of our policy learning approach converges to that of the oracle as the sample size increases. Moreover, when compared to the \addIPW{} estimator, the \addDR{} estimator provides additional variance reduction, further enhancing learning performance.
Among the IPW estimators, accounting for unknown interference (blue) improves performance but suffers from high estimation variance, making it substantially less effective than the proposed methodology. The anonymous IPW estimator (brown), while incorporating some interference effects, introduces bias by neglecting heterogeneous spillover effects in Equation~\eqref{equ:sim:DGP}. This results in the second-worst performance. As expected, the IPW estimator that ignores interference entirely (green) performs the worst.
Overall, the proposed methodology outperforms the existing IPW estimators.

The right panel of Figure~\ref{fig:sim:PL} shows the results under Scenario~B where Assumption~\ref{ass:additive} is violated and our \addIPW{} and \addDR{} estimators are biased. 
As shown in Equation~\eqref{equ:projection}, however, the proposed estimators still provide a reasonable approximation and capture a substantial proportion of the spillover effects.
We find that our approach remains superior to the existing IPW estimators due to a favorable bias-variance tradeoff. 
While the performance gaps between the proposed estimators and the IPW estimators with unknown and anonymous interference are smaller than in Scenario~A, our approach still achieves lower variance, with the \addDR{} estimator further reducing variance compared to the \addIPW{} estimator.
As before, the IPW estimators that account for interference significantly outperform the one that ignores it, highlighting the importance of modeling spillover effects in policy learning.

Efficient policy learning relies on accurate policy evaluation.
Next, we assess the finite-sample performance of the proposed estimators in estimating policy values.  We compare this performance with that of the IPW estimators that account for interference.
We use the same two simulation scenarios and consider the evaluation of three different policies: (i) linear policy: $\pi(\bX)=\mathds{1}(0.5X_1+X_2+X_3+0.5X_4\geq 2)$, (ii) depth-2 decision tree policy: $\pi(\bX)=\mathds{1}(X_1>0.5,X_3>0)$ and (iii) treat-nobody policy: $\pi(\bX)=0$.

\begin{figure}[!t]
    \centering \includegraphics[width=0.9\linewidth]{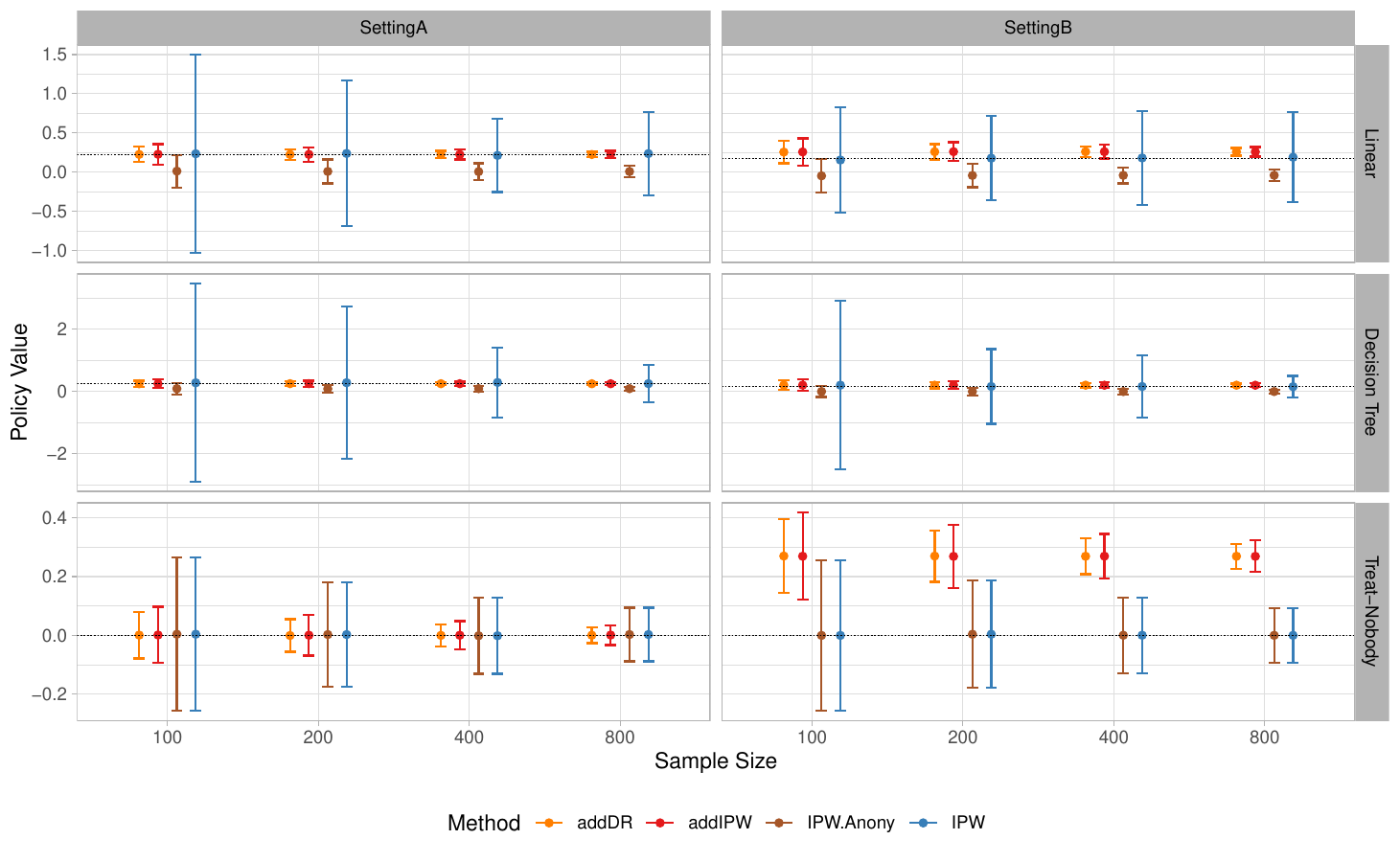}
    \caption{The performance of policy evaluation methodology based on the proposed \addIPW{} (red) and \addDR{} (orange) estimators, as well as the existing IPW estimators with unknown interference (blue) and anonymous interference (brown). The dots represent the average performance over simulations while the lines indicate the one standard deviation above and below the mean. The true policy value (dashed black line) is calculated based on Monte-Carlo simulations.}
    \label{fig:sim:PE}
\end{figure}

Figure~\ref{fig:sim:PE} shows the results, with each plot corresponding to a different combination of policy class (rows) and a data generation scenario (columns).
The results for the sample size of $n=50$ are omitted due to large variability.
Overall, the proposed \addIPW{} (red) and \addDR{} (orange) estimators exhibit a substantial advantage over the standard IPW estimators with unknown (blue) and anonymous (brown) interference in terms of standard deviation, across the settings we have considered.

In cases where Assumption~\ref{ass:additive} is satisfied (Scenario~A; left panel), our proposed estimators and the IPW estimator with unknown interference provide unbiased estimates of the policy value. However, the \addIPW{} and \addDR{} estimators exhibit considerably lower variability compared to the IPW estimator with unknown interference.
In contrast, when Assumption~\ref{ass:additive} is violated (Scenario~B; right panel), the proposed estimators may deviate from the true policy value due to model misspecification.
The degree of bias varies across policies with the most substantial bias arising for evaluating the treat-nobody policy.
Meanwhile, the IPW estimator with unknown interference remains unbiased in both scenarios. The anonymous IPW estimator, however, is biased in both settings, with variance between those of our proposed estimators and the IPW estimator with unknown interference. An exception occurs when evaluating the treat-nobody policy, where the anonymous IPW estimator exactly reduces to the IPW estimator with unknown interference.

In addition, the differences in variance between our proposed estimators and the existing IPW estimators are much greater for linear and decision-tree policies than for the treat-nobody policy.
This is because the variance of a policy value estimator depends on the deviation between the baseline policy (i.e., following propensity score) and the (deterministic) policy to be evaluated.

If the true optimal policy significantly deviates from the baseline policy, the IPW estimator with unknown interference can yield a highly variable estimate of the optimal policy value.
This in turn can negatively impact the downstream empirical policy optimization.
This bias-variance tradeoff is evident in our findings about policy learning (Figure~\ref{fig:sim:PL}), where learned policies based on the proposed estimators are more robust than those based on the existing IPW estimators, even when Assumption~\ref{ass:additive} is not met.

\section{Empirical Application}\label{sec:app}

We illustrate our methodology by applying it to a randomized experiment from a conditional cash transfer program in Colombia \citep{barrera2011improving}.

\subsection{Setup}

The experiment was conducted in two regions of Bogota, Colombia: San Cristobal and Suba. 
In each region, the researchers recruited households that have one to five schoolchildren, and within each household, the children were randomized to enroll in the cash transfer program. 
Specifically, the researchers stratified children based on locality (San Cristobal or Suba), type of school (public or private), gender, and grade level. 
While within each stratum, every child had an equal probability of receiving treatment, the treatment assignment probabilities varied across strata.
The original treatment randomization probability for each stratum is known, and is on average $0.63$ for children in San Cristobal and  $0.45$ for those in Suba.
Since randomization was based on fine strata that include gender and grade level, almost all children within each household belong to different strata.
Therefore, we can assume that the treatment assignment mechanisms for children are conditionally independent of one another and satisfy Assumption~\ref{ass:factorCPS}.

Previous studies focused on estimating the effects of the conditional cash transfer program on the attendance rate of students.
The program was designed such that enrolled students received cash subsidies if they attended school at least 80\% of the time in a given month.
For example, the original study estimated the spillover effects on the siblings of an enrolled student in the same household \citep{barrera2011improving}.
In our application, we analyze the same dataset (a total of 1010 households with 2129 students) as the one examined by \cite{park2022efficient} who developed and applied semiparametric efficient estimators to estimate the direct and spillover effects of the program on school attendance rates.

In contrast to these previous studies, however, we focus on learning an optimal individualized treatment rule to maximize the average household-level school attendance rate.
We consider the following linear policy class based on three pre-treatment variables: student’s grade, household's size, and household’s poverty score (with lower scores indicating poorer households):
\[\pi(\bX_{ij}) = \mathds{1}\{\beta_0+\beta_1\times\text{student's grade}+\beta_2\times \text{household's size}
+\beta_3 \times\text{household’s poverty score}\geq 0\}.
\]

In order to evaluate the performance of the learned individualized policy, we randomly split the data into $K=5$ folds $\mathcal{I}_k$, $k=1,\ldots,K$, and for each fold $k$, learn an optimal individualized treatment policy $\hat{\pi}^{(-k)}(\cdot)$ using all but the $k$-th fold data.
We incorporate the cost of treatment into the objective function to ensure that the optimal treatment rule is not trivial \citep{athey2021policy}.
We then evaluate the the learned policy using the $k$th fold, and compute the overall empirical policy value by averaging over the resulting five estimates.
Specifically, we use the following policy value estimator with a varying treatment cost $C$ ($15\%$, $20\%$ and $25\%$ of the school-attendance rate benefit, which ranges from 0 to 1).
\begin{equation}\label{equ:app:evalV}
  \widehat{V}= \frac{1}{n}\sum_{k=1}^{K}\sum_{i\in\mathcal{I}_k}\ \overline{\bY}_{i} \left[\sum_{j=1}^{M_i}
  \left( \frac{\mathds{1}\left\{
       A_{ij}=\hat{\pi}^{(-k)}(\bX_{ij})
     \right\}}{  e_j(
\hat{\pi}^{(-k)}(\bX_{ij})
      \mid \bX_{i})}
  -1\right)+1\right]-C\times \frac{\sum_{j=1}^{M_i}\hat{\pi}^{(-k)}(\bX_{ij})}{M_i}.  
\end{equation}

\subsection{Results}

\begin{table}[t!]
    \centering
\begin{tabular}{l|cccc}
\hline
\hline
 & \multicolumn{3}{|c}{ $\begin{array}{cc} \text{Estimated policy value}\\ \text{(Treated Proportion)} \end{array}$ } \\ 
 \cline{2-4}
Method & cost: $15\%$ &  cost: $20\%$ & cost: $25\%$ \\ 
\hline \addIPW{}  &   0.927 (0.542) &   0.904 (0.475) &  0.897 (0.478) \\
IPW with unknown interference &   0.927  (0.618) &  0.867 (0.485) &  0.831 (0.436) \\
IPW with anonymous interference &   0.905   (0.556) &   0.855 (0.455) &  0.828 (0.514) \\
IPW with no interference & 0.872 (0.359) &   0.847 (0.277) &  0.808 (0.124) \\
\hline
\end{tabular}  
\caption{Estimated values of the learned policies under different treatment costs. The number in parentheses represents the proportion of children assigned to the treatment.  The proposed individualized policy learning methodology (based on \addIPW{}) outperforms the learned policies based on the standard IPW estimators, yielding the highest policy values. }
\label{tab:estV}
\end{table}

Table~\ref{tab:estV} presents the estimated values of the learned individualized policies based on our proposed \addIPW{} estimator and those based on the standard IPW estimators with unknown, anonymous, and no interference. 
While both \addIPW{} and \addDR{} estimators are consistent when propensity scores are known, we exclude the \addDR{} estimator in this application due to poor fit of the outcome model, leading to a degenerate learned policy. 

We find that the proposed policy learning methodology consistently achieves the best performance across all treatment cost specifications, yielding the highest estimated policy values. In this example, since each cluster contains only a small number of units (mostly 2 or 3), the standard IPW estimator with unknown interference does not suffer from a significantly large estimation variance and performs reasonably well. As a result, its estimated policy value and the proportion of treated units (reported in parentheses) are relatively close to those of our proposed method. The IPW estimator with anonymous interference performs slightly worse, as it fails to account for heterogeneous spillover effects within clusters. In contrast, the learned policy based on the standard IPW estimator that ignores interference performs the worst in all cases, yielding the lowest policy values and treating the smallest proportion of units.

\begin{table}[t]
    \centering
    \begin{tabular}{lrrr}
\hline \hline
 & grade level & household size & poverty score \\
\hline
cost: $15\%$ &   \\
\addIPW{} &   $-$3.914 & 0.721 & $-$3.425 \\ 
IPW with unknown interference  &  $-$0.977 & 0.754 & $-$0.515 \\ 
IPW with anonymous interference & $-$1.040 & 0.972 & $-$0.790 \\ 
IPW with no interference  &  $-$1.918 & 2.376 & $-$0.884 \\ \\ 
cost: $20\%$ &   \\ 
\addIPW{}  &  $-$4.009 & 0.744 & $-$3.484 \\ 
IPW with unknown interference  & $-$3.974 & 0.769 & $-$3.483 \\ 
IPW with anonymous interference & $-$1.039 & 0.972 & $-$0.789 \\ 
IPW with no interference  & 0.523 & 0.166 & $-$0.816 \\ \\
cost: $25\%$ &   \\
\addIPW{}  & $-$1.921 & 2.441 & $-$0.876 \\ 
IPW with unknown interference & $-$2.303 & 2.842 & $-$0.265 \\ 
IPW with anonymous interference & $-$3.985 & 0.751 & $-$3.466 \\ 
IPW with no interference  & 0.411 & 0.273 & $-$0.598 \\ 
\hline
\end{tabular}
 \caption{Estimated coefficients of the linear optimal policy under all methodologies and all treatment cost specifications, normalized by the absolute magnitude of intercept. The learned policies tend to prioritize children who are in lower grade levels and live in larger and poorer households. 
} \label{tab:coef}
\end{table}

We also report the estimated coefficients of the linear policies in Table~\ref{tab:coef}, normalized by the absolute magnitude of the estimated intercept $|\hat\beta_0|$.
Here, we use the entire data to learn a single policy under each methodology.
We find that, regardless of costs, the optimal ITR tends to negatively depend on the student's grade and the household's poverty score, while it has a positive relationship with household's size across all methodologies. 
This implies that optimal policies tend to give priority to children who are in a lower grade level and reside in larger and poorer households.

However, the estimated optimal policy based on the standard IPW estimator without interference can yield qualitatively different results.
For example, when the cost is 20\% and 25\%, the estimated coefficients for the grade level become positive, implying the opposite treatment recommendation rule.
Furthermore, although the standard IPW estimator with interference yields the estimated coefficients of the same sign as those based on our methodology, its statistical inefficiency may have produced lower policy values, as shown in Table~\ref{tab:estV}.

\section{Concluding Remarks}\label{sec:conclusion}

In this paper, we propose new policy evaluation and learning methodologies under clustered network interference.
We introduce the semiparametric additive effect model that flexibly captures heterogeneous spillover effects among individuals while avoiding restrictive assumptions used in the existing literature.
Our proposed \addIPW{} and \addDR{} estimators for policy evaluation exploit this structural assumption and substantially improve statistical efficiency relative to the standard IPW estimator.
Theoretically, we establish regret bounds for the learned ITR that achieve optimal rates with respect to sample size, under both known and unknown propensity score settings.

The empirical results demonstrate the importance of considering individual-level network information in the policy learning problem.
Consistent with our theoretical analysis, we find that even when our assumption is violated, the proposed estimators outperform the standard IPW estimators.
Thus, the proposed methodology achieves a desirable degree of bias-variance tradeoff by leveraging a structural assumption that is sufficiently informative but is not too restrictive.

An interesting future direction is to extend our methodology to more general network settings, such as incorporating weak dependencies between clusters or adapting it to a single-network setup.
Our proposed semiparametric additive model could be modified to accommodate these interference structures, relaxing the current restrictive assumptions on spillover effects.
While these extensions may require a different theoretical analysis framework, they would further enhance the applicability of individualized policy learning in real-world settings with complex interdependence.

\bigskip

\bibliographystyle{chicago}
\bibliography{ref.bib}

\newpage
\appendix
\renewcommand\thefigure{\thesection.\arabic{figure}}    
\setcounter{figure}{0}

\begin{center}
 \huge Supplementary Appendix 
\end{center}

\etocdepthtag.toc{mtappendix}
\etocsettagdepth{mtchapter}{none}
\etocsettagdepth{mtappendix}{subsection}
\tableofcontents

\section{Proofs}\label{appendix:proof}

\subsection{Proof of Proposition~\ref{prop:unbiased}}\label{appendix:proof:unbiasedness}
\begin{proof}
For notational simplicity, we define the cluster-level average nuisance functions of Equation~\eqref{equ:additive} as $$\overline{\bg}(\bX_i)=\frac{1}{M_i}\sum_{j=1}^{M_i}\bg_{j}(\bX_i):=\left(\overline{g}^{(0)}(\bX_i),\overline{g}^{(1)}(\bX_i),\ldots,\overline{g}^{(M_i)}(\bX_i)\right)^\top.$$
We start by examining the expectation of the proposed estimator,
\begin{equation}
\begin{aligned}\label{equ:prf0}
    \mathbb{E}[ \widehat{V}^{\addIPW}(\pi) ]=& \mathbb{E}\left[ \frac{1}{n}\sum_{i=1}^{n}\overline{\bY}_{i}\left\{ \sum_{j=1}^{M_i}\left(
   \frac{\mathds{1}\left\{
       A_{ij}=\pi(\bX_{ij})
     \right\}}{  e_j(
 \pi(\bX_{ij})
      \mid \bX_{i})}
  -1\right)+1\right\}\right]\\
  =& \mathbb{E}\left[ \frac{1}{n}\sum_{i=1}^{n} \mathbb{E} \left[\overline{\bY}_{i}\sum_{j=1}^{M_i}\left( 
   \frac{\mathds{1}\left\{
       A_{ij}=\pi(\bX_{ij})
     \right\}}{  e_j(
 \pi(\bX_{ij})
      \mid\bX_{i})}
  -1\right) \ \Biggl | \ \bX_i\right] +\mathbb{E} \left[\overline{\bY}_{i} \mid \bX_i\right]\right].
\end{aligned}
\end{equation}
We can further compute, 
\begin{align}
 & \mathbb{E} \left[\overline{\bY}_{i}
   \frac{\mathds{1}\left\{
       A_{ij}=\pi(\bX_{ij})
     \right\}}{  e_j(
 \pi(\bX_{ij})
      \mid\bX_{i})}
 \ \Biggl | \ \bX_i\right]  \nonumber\\
 = \ &  \mathbb{E} \left[ \overline{\bY}_{i}(\bA_{i(-j)},A_{ij}=\pi(\bX_{ij}))
   \frac{\mathds{1}\left\{
       A_{ij}=\pi(\bX_{ij})
     \right\}}{  e_j(
 \pi(\bX_{ij})
      \mid \bX_{i})}
 \ \Bigl | \ \bX_i\right] \nonumber\\
 = \ & \mathbb{E} \left[ \overline{\bY}_{i}(\bA_{i(-j)},A_{ij}=\pi(\bX_{ij}))  \mid \bX_i\right]
  \mathbb{E} \left[  \frac{\mathds{1}\left\{
       A_{ij}=\pi(\bX_{ij})
     \right\}}{  e_j(
 \pi(\bX_{ij})
      \mid\bX_{i})}
 \ \Bigl | \ \bX_i\right] \nonumber\\
= \ & \mathbb{E} \left[ \E \left[  \overline{\bY}_{i}(\bA_{i(-j)},A_{ij}=\pi(\bX_{ij})) \mid \bA_{i(-j)}, A_{ij}=\pi(\bX_{ij}),\bX_i \right] \mid \bX_i\right]\nonumber\\
= \ & \overline{g}^{(0)}(\bX_i)+\sum_{k=1,k\neq j}^{M_i}\overline{g}^{(k)}(\bX_i)e_k(1\mid \bX_i)+\overline{g}^{(j)}(\bX_i)\pi(\bX_{ij}), \label{equ:prf1}
\end{align}
where the first equality is due to consistency, the second equality follows from unconfoundedness (Assumption~\ref{ass:DGP}(a)) and factorization of propensity scores (Assumption~\ref{ass:factorCPS}),
the third equality follows from the tower property and the fact that $\bY_i\left(\ba_i\right) \indep A_{ij} \mid \bA_{i(-j)},\bX_i$ (due to Assumptions~\ref{ass:DGP}(a)~and~\ref{ass:factorCPS}) and the final equality holds under the additive outcome model assumption (Assumption~\ref{ass:additive}).
In addition, Assumption~\ref{ass:additive} implies the following equality,
\begin{equation}
\begin{aligned}\label{equ:prf2}
\mathbb{E} \left[\overline{\bY}_{i} \mid \bX_i\right] &=  \overline{g}^{(0)}(\bX_i)+\sum_{k=1}^{M_i}\overline{g}^{(k)}(\bX_i)e_k(1\mid \bX_i).
\end{aligned}
\end{equation}
Plugging Equations~\eqref{equ:prf1}~and\eqref{equ:prf2} into 
Equation~\eqref{equ:prf0}, we obtain the following desired result,
\begin{equation*}
\begin{aligned}
  \mathbb{E}[ \widehat{V}^{\addIPW}(\pi) ] \
 = \ &\mathbb{E}\left[  M_i\times \overline{g}^{(0)}(\bX_i)+(M_i-1)\sum_{k=1}^{M_i}\overline{g}^{(k)}(\bX_i)e_k(1\mid \bX_i)+\sum_{j=1}^{M_i}\overline{g}^{(j)}(\bX_i)\pi(\bX_{ij})\right.\\
 &\quad\quad\left.-(M_i-1)\left(\overline{g}^{(0)}(\bX_i)+\sum_{k=1}^{M_i}\overline{g}^{(k)}(\bX_i)e_k(1\mid \bX_i) \right)\right]\\
=\ & \mathbb{E}\left[\overline{g}^{(0)}(\bX_i)+\sum_{j=1}^{M_i}\overline{g}^{(j)}(\bX_i)\pi(\bX_{ij})\right]\\
=\ & V(\pi).
\end{aligned}
\end{equation*}
where the last equality uses the law of iterated expectation under the modeling assumption (Assumption~\ref{ass:additive}).
\end{proof}

\subsection{Proof of Theorem~\ref{thm:regret}}\label{appendix:proof:finiteregretbound}

\begin{definition}[Rademacher complexity]
Let $\mathcal{F}$ be a family of functions mapping $\mathcal{O} \mapsto \mathbb{R}$. The population Rademacher complexity of $\mathcal{F}$ is defined as
$$\mathcal{R}_{n}(\mathcal{F}) \equiv \mathbb{E}_{O, \varepsilon}\left[\sup _{f \in \mathcal{F}}\left|\frac{1}{n} \sum_{i=1}^{n} \varepsilon_{i} f\left(O_{i}\right)\right|\right]$$
where $\varepsilon_{i}$'s are i.i.d. Rademacher random variables, i.e., $\operatorname{Pr}\left(\varepsilon_{i}=1\right)=\operatorname{Pr}\left(\varepsilon_{i}=-1\right)=1 / 2$, and the expectation is taken over both the Rademacher variables $\varepsilon_{i}$ and the i.i.d. random variables $O_{i}$. 
\end{definition}

\begin{proof}
    By definition, the regret of $\pihat$ relative to $\piast$ is given by
    \begin{equation*}
        \begin{aligned}
             R(\pihat)=V(\piast)-V(\pihat)&=V(\piast)-\widehat{V}(\piast)+\widehat{V}(\piast)-\widehat{V}(\pihat)+\widehat{V}(\pihat)-V(\pihat)\\
             &\leq V(\piast)-\widehat{V}(\piast)+\widehat{V}(\pihat)-V(\pihat)\\
             &\leq 2\sup_{\pi\in\Pi}\lvert \widehat{V}(\pi)-V(\pi) \rvert.
        \end{aligned}
    \end{equation*}
Similar to the standard policy learning literature, our approach is to bound
the above worst-case estimation error, $\sup_{\pi\in\Pi}\lvert \widehat{V}(\pi)-V(\pi) \rvert$.
Define
\begin{equation*}
   f(O_i;\pi)=\left[\overline{\bY}_{i} \left\{\sum_{j=1}^{M_i}
  \left( \frac{\mathds{1}\left\{
       A_{ij}=\pi(\bX_{ij})
     \right\}}{  e_j(
 \pi(\bX_{ij})
      \mid\bX_{i})}
  -1\right)+1\right\} \right],
\end{equation*}
and the function class on $\mathcal{O}$ as $\mathcal{F}=\left\{ f(\cdot;\pi) \mid \pi\in\Pi \right\}$.
Proposition~\ref{prop:unbiased} implies,
$$
\sup _{\pi \in \Pi}|\widehat{V}(\pi)-V(\pi)|=\sup _{f \in \mathcal{F}}\left|\frac{1}{n} \sum_{i=1}^n f\left(O_i\right)-\mathbb{E}[f(O)]\right|.
$$
From Assumptions~\ref{ass:factorCPS}~and~\ref{ass:regret:bound}, the class $\mathcal{F}$ is uniformly bounded by a constant $C=B[m_{\max}\times(\frac{1}{\eta}-1)+1]$. Then by Theorem 4.10 in \cite{wainwright2019high},
$$
\sup _{f \in \mathcal{F}}\left|\frac{1}{n} \sum_{i=1}^n f\left(O_i\right)-\mathbb{E}[f(O)]\right| \leq 2 \mathcal{R}_n(\mathcal{F})+\delta
$$
with probability at least $1-\exp \left(-\frac{n \delta^2}{2 C^2}\right)$.
It suffices to bound the Rademacher complexity for $\mathcal{F}$, $\mathcal{R}_n(\mathcal{F})$, given by 
$$
\begin{aligned}
\mathcal{R}_{n}(\mathcal{F}) \ \leq \ & \mathbb{E}_{O, \epsilon}\left[\left| \frac{1}{n} \sum_{i=1}^{n} \epsilon_{i} \overline{\bY}_{i} \left\{\sum_{j=1}^{M_i}
  \left(\frac{1-A_{ij}}{  e_j(0\mid\bX_{i})}
  -1\right)+1\right\} \right|\right] \\
& +  \mathbb{E}_{O, \varepsilon}\left[\sup _{\pi \in \Pi}\left| \frac{1}{n} \sum_{i=1}^{n}  \epsilon_{i}\overline{\bY}_{i} \sum_{j=1}^{M_i} \pi(\bX_{ij})\left(\frac{A_{ij}}{ e_j(1\mid\bX_{i})}- \frac{1-A_{ij}}{ e_j(0\mid\bX_{i})}\right) \right|\right] \\
\ \leq \ & \frac{C}{n} \sqrt{n}
+\frac{B}{\eta} \mathbb{E}_{\bX,M, \varepsilon}\left[\sup _{\pi \in \Pi}\left|\frac{1}{n} \sum_{i=1}^{n}  \varepsilon_{i}\sum_{j=1}^{M_i} \pi(\bX_{ij})\right|\right].\\
\ \leq \ & \frac{C}{\sqrt{n}}
+\frac{B}{\eta} \mathbb{E}_{\bX,M, \varepsilon}\left[\sup _{\pi \in \Pi}\left|\frac{1}{n} \sum_{i=1}^{n}  \varepsilon_{i}\sum_{j=1}^{M_i} \pi(\bX_{ij})\right|\right],\\
\end{aligned}
$$
where the second inequality follows from the application of the Khintchine inequality to the first term and that of the contraction inequality for Rademacher complexities to the second term, respectively.

To bound the second term above, we utilize the upper bound $m_{\max }$ on cluster size given in Assumption~\ref{ass:regret:bound}(b). We append zeros to the cluster-level covariates $\bX_{i}$ (i.e., $\bX_{ij}=\emptyset$ for $j=M_i+1,\ldots,m_{\max}$) so that each cluster has the same length of covariate vectors, and define $\pi(\emptyset)=0$.
Therefore, we can bound it by
\begin{equation*}
    \begin{aligned}
     \mathbb{E}_{\bX,M, \varepsilon}\left[   \sup _{\pi \in \Pi}\left|\frac{1}{n} \sum_{i=1}^{n}  \varepsilon_{i}\sum_{j=1}^{m_{\max}} \pi(\bX_{ij})\right|\right]
     \leq & \mathbb{E}_{\bX,M, \varepsilon}\left[\sup _{\pi \in \Pi} \sum_{j=1}^{m_{\max}} \left|\frac{1}{n} \sum_{i=1}^{n}  \varepsilon_{i}\pi(\bX_{ij})\right|\right]\\
     \leq & \mathbb{E}_{\bX,M, \varepsilon}\left[\sum_{j=1}^{m_{\max}}\sup _{\pi \in \Pi}  \left|\frac{1}{n} \sum_{i=1}^{n}  \varepsilon_{i}\pi(\bX_{ij})\right|\right]\\
       \leq&\sum_{j=1}^{m_{\max}}\mathbb{E}_{\bX,M, \varepsilon}\left[\sup _{\pi \in \Pi} \left|\frac{1}{n} \sum_{i=1}^{n}  \varepsilon_{i}\pi(\bX_{ij})\right|\right]\\
       \leq &  c_0 m_{\max}\sqrt{\frac{\nu}{n}},
    \end{aligned}
\end{equation*}
where the first inequality is due to triangle inequality,
and the last inequality follows from the fact that for a function class $\mathcal{G}$ with finite VC dimension $\nu<\infty$, the Rademacher complexity is bounded by, $\mathcal{R}_{n}(\mathcal{G}) \leq c_0 \sqrt{\nu / n}$, for some universal constant $c_0$ \citep[][$\S 5.3.3$]{wainwright2019high}. 
Combining the above, we obtain $\mathcal{R}_{n}(\mathcal{F})\leq \frac{C}{\sqrt{n}}+c_0\frac{Bm_{\max}}{\eta}\sqrt{\frac{\nu}{n}}$. This completes the proof.
\end{proof}

\subsection{Proof of Theorem~\ref{thm:semieff}}\label{appendix:pf:semipara}
\begin{proof}

We first derive the semiparametric efficiency bound for a policy value estimate. We use a result from \cite{chamberlain1992efficiency} who
derives the efficiency bound under conditional moment restrictions with a unknown nonparametric component. 
Assumption~\ref{ass:additive} can be written as the following conditional moment restriction,
\[
 \E\left[\bY-G(\bX) \phi(\bA)\mid \bA,\bX\right]=\mathbf{0}.
\]
Based on the results in \cite{chamberlain1992efficiency}, we can directly show that  under the additive semiparametric model assumption, the efficiency bound  for estimating $V(\pi)$ is:
 \begin{equation*}
        \begin{aligned}
         &\Var\left[w(M)^\top G(\bX)\phi(\pi(\bX))\right]+\\
         &\E\left[ \phi(\pi(\bX))^\top \left( \E\left[\phi(\bA) \Var\left[ w(M)^\top \left(\bY-G(\bX)\phi(\bA) \right) \mid \bA,\bX \right]^{-1}\phi(\bA)^\top \mid \bX \right] \right)^{-1} \phi(\pi(\bX)) \right]\\
        \end{aligned}
\end{equation*}
Under the homoskedastivity assumption, $\mathbb{E}\left[w(M)^\top(\bY- \mu(\bA,\bX) )(\bY- \mu(\bA,\bX) )^\top w(M) \mid \bA,\bX \right]=\sigma^2$ where $\sigma^2$ is a constant.  Thus, this efficiency bound becomes:
\begin{equation}\label{equ:semi:effbound}
    \begin{aligned}
         &\Var\left[w(M)^\top G(\bX)\phi(\pi(\bX))\right]+\E\left[ \sigma^2 \phi(\pi(\bX))^\top \left( \E\left[\phi(\bA) \phi(\bA)^\top \mid \bX \right] \right)^{-1} \phi(\pi(\bX)) \right]\\
          =&\Var\left[w(M)^\top G(\bX)\phi(\pi(\bX))\right]+\E\left[ \sigma^2 \phi(\pi(\bX))^\top \Sigma(\bX)^{-1} \phi(\pi(\bX)) \right].
    \end{aligned}
\end{equation}

Next, we consider the oracle doubly robust estimator using the true nuisance parameters $G(\bX)$ and $\Sigma(\bX)$, denoted as $\widetilde{V}^{\addDR}(\pi)$. We first show that the efficiency bound derived above in Equation~\eqref{equ:semi:effbound} is equal to the asymptotic variance of the oracle doubly robust estimator:
\begin{equation*}
    \begin{aligned}
        &\Var\left[w(M)^\top G_\addDR(\bY,\bX,M)\phi(\pi(\bX))\right]\\
        = &\Var\left[w(M)^\top G(\bX)\phi(\pi(\bX))\right]+ \Var\left[ w(M)^\top \left(\bY-G(\bX)\phi(\bA) \right)\phi(\bA)^\top \Sigma(\bX)^{-1} \phi(\pi(\bX)) \right]\\
        =& \Var\left[w(M)^\top G(\bX)\phi(\pi(\bX))\right]+\\
        &\E\left[\Var\left[ w(M)^\top \left(\bY-G(\bX)^\top\phi(\bA) \right)\mid  \bA,\bX \right] (\phi(\bA)^\top \Sigma(\bX)^{-1} \phi(\pi(\bX)))^2\right]\\
        =& \Var\left[w(M)^\top G(\bX)^\top\phi(\pi(\bX))\right]+\\    &\E\left[\sigma^2\phi(\pi(\bX))^\top\Sigma(\bX)^{-1} \E[\phi(\bA)\phi(\bA)^\top \mid \bX,M]\Sigma(\bX)^{-1} \phi(\pi(\bX))\right]\\
        =& \Var\left[w(M)^\top G(\bX)^\top\phi(\pi(\bX))\right]+\E\left[\sigma^2 \phi(\pi(\bX))^\top\Sigma(\bX)^{-1}  \phi(\pi(\bX))\right].\\
    \end{aligned}
\end{equation*}
 
Thus, it suffices to show  $\widetilde{V}^{\addDR}(\pi)-\widehat{V}^{\addDR}(\pi)=o_p(1/\sqrt{n})$.
We have the following decomposition:
\begin{equation*}
    \begin{aligned}
        \widetilde{V}^{\addDR}(\pi)-\widehat{V}^{\addDR}(\pi)  = &  \frac{1}{n}\sum_{i=1}^{n} \left(\hat{\overline{\bg}}(\bX_i)-\overline{\bg}(\bX_i)\right)^\top \left(\phi(\bA_i)\phi(\bA_i)^\top\Sigma(\bX)^{-1} -1\right)\phi(\pi(\bX_i))\\
        & +\frac{1}{n}\sum_{i=1}^{n}  \left(\overline{\bY}_{i}- \overline{\bg}^\top(\bX_i)\phi(\bA_i)\right)\phi(\bA_i)^\top \left( {\Sigma}(\bX_i)^{-1} -\widehat{\Sigma}(\bX_i)^{-1} \right)\phi(\pi(\bX_i)) \\
        & + \frac{1}{n}\sum_{i=1}^{n} \left(\hat{\overline{\bg}}(\bX_i)-\overline{\bg}(\bX_i)\right)^\top \phi(\bA_i)\phi(\bA_i)^\top  \left( \widehat{\Sigma}(\bX_i)^{-1} -{\Sigma}(\bX_i)^{-1} \right)\phi(\pi(\bX_i)).        
    \end{aligned}
\end{equation*}

We study each terms in the above expression separately.  Since clusters are i.i.d, the first term has mean zero:
\begin{equation*}
    \begin{aligned}
         & \E\left[\left(\hat{\overline{\bg}}(\bX_i)-\overline{\bg}(\bX_i)\right)^\top \left(\phi(\bA_i)\phi(\bA_i)^\top\Sigma(\bX)^{-1} -1\right)\phi(\pi(\bX_i))\right] \\
         =& \E\left[\left(\hat{\overline{\bg}}(\bX_i)-\overline{\bg}(\bX_i)\right)^\top \left(\Sigma(\bX)\Sigma(\bX)^{-1} -1\right)\phi(\pi(\bX_i))\right] \\
         = & 0.
    \end{aligned}
\end{equation*}
and its second moment
\begin{equation*}
    \begin{aligned}
           &\  \E\left[\left(\frac{1}{n}\sum_{i=1}^{n} \left(\hat{\overline{\bg}}(\bX_i)-\overline{\bg}(\bX_i)\right)^\top \left(\phi(\bA_i)\phi(\bA_i)^\top\Sigma(\bX)^{-1} -1\right)\phi(\pi(\bX_i))\right)^2\right] \\
           = & \ \Var\left[\frac{1}{n}\sum_{i=1}^{n} \left(\hat{\overline{\bg}}(\bX_i)-\overline{\bg}(\bX_i)\right)^\top \left(\phi(\bA_i)\phi(\bA_i)^\top\Sigma(\bX)^{-1} -1\right)\phi(\pi(\bX_i))\right] \\
           = & \ \frac{1}{n} \Var\left[\left(\hat{\overline{\bg}}(\bX_i)-\overline{\bg}(\bX_i)\right)^\top \left(\phi(\bA_i)\phi(\bA_i)^\top\Sigma(\bX)^{-1} -1\right)\phi(\pi(\bX_i))\right] \\
           = & \  \frac{1}{n} \E\left[\left( \left(\hat{\overline{\bg}}(\bX_i)-\overline{\bg}(\bX_i)\right)^\top \left(\phi(\bA_i)\phi(\bA_i)^\top\Sigma(\bX)^{-1} -1\right)\phi(\pi(\bX_i))\right)^2\right] \\
           \lesssim &\ \frac{1}{n} \E\left[\left(\hat{\overline{\bg}}(\bX_i)-\overline{\bg}(\bX_i)\right)^\top\left(\hat{\overline{\bg}}(\bX_i)-\overline{\bg}(\bX_i)\right) \right] \\
           = & \ \frac{ O_{p}\left(r_{n,g}^2\right)}{n} = \frac{o_p(1)}{n}
    \end{aligned}
\end{equation*}
where the second-to-last line is due to the propensity scores being strictly bounded away from zero and one as well as the Cauchy-Schwarz inequality, and the last line is due to Assumption~\ref{ass:rate}.
By Chebyshev's inequality, the first term can be bounded by $o_p(1/\sqrt{n})$.
The second term in our decomposition can also be bounded similarly.
Finally, for the last term, we simply use Cauchy-Schwarz: \allowdisplaybreaks
\begin{equation*} 
\begin{aligned} 
           & \E\left[\left\lvert\frac{1}{n}\sum_{i=1}^{n} \left(\hat{\overline{\bg}}(\bX_i)-\overline{\bg}(\bX_i)\right)^\top \phi(\bA_i)\phi(\bA_i)^\top  \left( \widehat{\Sigma}(\bX_i)^{-1} -{\Sigma}(\bX_i)^{-1} \right)\phi(\pi(\bX_i))\right\rvert\right] \\
           \leq & \ \E\left[ \left\|\frac{1}{\sqrt{n}}\sum_{i=1}^{n}\left(\hat{\overline{\bg}}(\bX_i)-\overline{\bg}(\bX_i)\right)^\top\phi(\bA_i)\right\|_2 \left\|\frac{1}{\sqrt{n}}\sum_{i=1}^{n} \phi(\bA_i)^\top  \left( \widehat{\Sigma}(\bX_i)^{-1} -{\Sigma}(\bX_i)^{-1} \right)\phi(\pi(\bX_i))\right\|_2\right] \\
           \leq & \ \sqrt{\E\left[ \frac{1}{n}\sum_{i=1}^{n}\left(\left(\hat{\overline{\bg}}(\bX_i)-\overline{\bg}(\bX_i)\right)^\top\phi(\bA_i)\right)^2\right]}\times  \sqrt{\E\left[ \frac{1}{n}\sum_{i=1}^{n} \left( \phi(\bA_i)^\top  \left( \widehat{\Sigma}(\bX_i)^{-1} -{\Sigma}(\bX_i)^{-1} \right)\phi(\pi(\bX_i))  \right)^2\right] }\\
           = & \ \sqrt{\E\left[\left(\left(\hat{\overline{\bg}}(\bX_i)-\overline{\bg}(\bX_i)\right)^\top\phi(\bA_i)\right)^2\right]}\times  \sqrt{\E\left[ \left( \phi(\bA_i)^\top  \left( \widehat{\Sigma}(\bX_i)^{-1} -{\Sigma}(\bX_i)^{-1} \right)\phi(\pi(\bX_i))  \right)^2\right] }\\
          =  & \sqrt{O_p(r_{n,g}^2)} \sqrt{O_p(r_{n,\Sigma}^2)} = o_p\left(\frac{1}{\sqrt{n}}\right)
    \end{aligned}
\end{equation*} 
by Assumption~\ref{ass:rate}. Here, we establish the result by leveraging the implicit regularity condition that the matrix \( {\Sigma}(\bX) \) is nonsingular and well-conditioned, and that \( \widehat{\Sigma}(\bX) \) converges to \( {\Sigma}(\bX) \) at a sufficient rate to ensure \( \widehat{\Sigma}(\bX) \) remains nonsingular. Consequently, \( \widehat{\Sigma}(\bX_i)^{-1} - {\Sigma}(\bX_i)^{-1} \) converges at the same rate as \( \widehat{\Sigma}(\bX_i) - {\Sigma}(\bX_i) \), as provided in  Assumption~\ref{ass:rate}.  Summing up these three bounds completes the proof.
\end{proof}

\subsection{Proof of Theorem~\ref{thm:regret:DR}}\label{appendix:pf:DRregret}
\begin{proof}

To proceed with the proof, we leverage the results in \cite{zhou2023offline}. We first adopt several useful definitions.
\begin{definition}
Given the individual-level feature domain $\mathcal{X}$, a policy class $\Pi$, a set of $n$ points $\left\{x_{1}, \ldots, x_{n}\right\} \subset$ $\mathcal{X}$, define:
\begin{enumerate}
\item Hamming distance between any two policies $\pi_a$ and $\pi_b$ in $\Pi: H\left(\pi_a, \pi_b\right)=\frac{1}{n} \sum_{i=1}^{n} \mathds{1}\left(\pi_a\left(x_{i}\right) \neq \pi_b\left(x_{i}\right)\right)$.
    
\item $\epsilon$-Hamming covering number of the set $\left\{x_{1}, \ldots, x_{n}\right\}$ :
$N_{H}\left(\epsilon, \Pi,\left\{x_{1}, \ldots, x_{n}\right\}\right)$ is the smallest number $K$ of policies $\left\{\pi_1, \ldots, \pi_{K}\right\}$ in $\Pi$, such that $\forall \pi \in$ $\Pi, \exists \pi_{i}, H\left(\pi, \pi_{i}\right) \leq \epsilon .$

\item $\epsilon$-Hamming covering number of $\Pi: N_{H}(\epsilon, \Pi)=\sup \left\{N_{H}\left(\epsilon, \Pi,\left\{x_{1}, \ldots, x_{m}\right\}\right) \mid m \geq 1, x_{1}, \ldots, x_{m} \in\right.$ $\mathcal{X}\} .$

\item Entropy integral: $\kappa(\Pi)=\int_{0}^{1} \sqrt{\log N_{H}\left(\epsilon^{2}, \Pi\right)} d \epsilon$.
\end{enumerate}
\end{definition}

\begin{definition}
Define the doubly robust estimator using simplified notation:
\begin{equation}
    \begin{aligned}
        \widehat{V}^{\addDR}(\pi) &= \frac{1}{n}\sum_{i=1}^{n} \left(\ \hat{\overline{\bg}}(\bX_i) +  \widehat{\Sigma}(\bX_i)^{-1}\left(\overline{\bY}_{i}- \hat{\overline{\bg}}^\top(\bX_i)\phi(\bA_i)\right)\phi(\bA_i) \right)^\top \phi(\pi(\bX_i)) \\
        & = \frac{1}{n}\sum_{i=1}^{n}\sum_{j=0}^{M_i} \left[\ \hat{\overline{\bg}}(\bX_i) +  \widehat{\Sigma}(\bX_i)^{-1}\left(\overline{\bY}_{i}- \hat{\overline{\bg}}^\top(\bX_i)\phi(\bA_i)\right)\phi(\bA_i) \right]_j \pi(\bX_{ij}) \\
        & = \sum_{j=0}^{m_{\max}} \underbrace{\frac{1}{n}\sum_{i=1}^{n}\left[\ \hat{\overline{\bg}}(\bX_i) +  \widehat{\Sigma}(\bX_i)^{-1}\left(\overline{\bY}_{i}- \hat{\overline{\bg}}^\top(\bX_i)\phi(\bA_i)\right)\phi(\bA_i) \right]_j \pi(\bX_{ij})}_{:=\widehat{V}^{\addDR}_j(\pi)}
    \end{aligned}
\end{equation}
By definition, $\phi(\pi(\bX_i))=(1,\pi(\bX_{i1}),\ldots,\pi(\bX_{iM_i}))$ is an $(M_i+1)$-dimensional vector. Following the same approach as in the proof of Theorem~\ref{thm:regret}, we append zeros to the policy vector by setting $\bX_{ij}=\emptyset$ for $j=M_i+1,\ldots,m_{\max}$ and defining $\pi(\emptyset)=0$. Additionally, we define the intercept term using the notation $\pi(\bX_{i0})\equiv1$.

Similarly, we define the oracle doubly robust estimator as
\begin{equation}
    \begin{aligned}
        \widetilde{V}^{\addDR}(\pi) &= \frac{1}{n}\sum_{i=1}^{n} \left(\ {\overline{\bg}}(\bX_i) +  {\Sigma}(\bX_i)^{-1}\left(\overline{\bY}_{i}- {\overline{\bg}}^\top(\bX_i)\phi(\bA_i)\right)\phi(\bA_i) \right)^\top \phi(\pi(\bX_i)) \\
        & = \sum_{j=0}^{m_{\max}} \underbrace{\frac{1}{n}\sum_{i=1}^{n}\left[\ {\overline{\bg}}(\bX_i) +  {\Sigma}(\bX_i)^{-1}\left(\overline{\bY}_{i}- {\overline{\bg}}^\top(\bX_i)\phi(\bA_i)\right)\phi(\bA_i) \right]_j \pi(\bX_{ij})}_{:=\widetilde{V}^{\addDR}_j(\pi)}
    \end{aligned}
\end{equation}
and the true policy value as 
\begin{equation}
    \begin{aligned}
        V(\pi) &= \E\left[ \left(\ {\overline{\bg}}(\bX_i) +  {\Sigma}(\bX_i)^{-1}\left(\overline{\bY}_{i}- {\overline{\bg}}^\top(\bX_i)\phi(\bA_i)\right)\phi(\bA_i) \right)^\top \phi(\pi(\bX_i))\right] \\
        & = \sum_{j=0}^{m_{\max}} \underbrace{\E\left[\left[\ {\overline{\bg}}(\bX_i) +  {\Sigma}(\bX_i)^{-1}\left(\overline{\bY}_{i}- {\overline{\bg}}^\top(\bX_i)\phi(\bA_i)\right)\phi(\bA_i) \right]_j \pi(\bX_{ij})\right]}_{:=V_j(\pi)}
    \end{aligned}
\end{equation}
\end{definition}

\begin{definition}\label{def:difffunction}
Define the following difference functions $\widetilde{\Delta}(\cdot , \cdot),\widehat{\Delta}(\cdot , \cdot),{\Delta}(\cdot , \cdot):\Pi \times \Pi \rightarrow \mathbb{R}$ between any two fixed policies $\pi_a$ and $\pi_b$ in $\Pi$ as:
$$\begin{aligned}
    \widetilde{\Delta}\left(\pi_a, \pi_b\right)&:=   \widetilde{V}^{\addDR}(\pi_a)- \widetilde{V}^{\addDR}(\pi_b)=\sum_{j=0}^{m_{\max}} \underbrace{\widetilde{V}^{\addDR}_j(\pi_a)-\widetilde{V}^{\addDR}_j(\pi_b)}_{ :=\widetilde{\Delta}_j\left(\pi_a, \pi_b\right)}\\
     \widehat{\Delta}\left(\pi_a, \pi_b\right)&:=   \widehat{V}^{\addDR}(\pi_a)- \widehat{V}^{\addDR}(\pi_b)=\sum_{j=0}^{m_{\max}} \underbrace{\widehat{V}^{\addDR}_j(\pi_a)-\widehat{V}^{\addDR}_j(\pi_b)}_{ :=\widehat{\Delta}_j\left(\pi_a, \pi_b\right)}\\
    {\Delta}\left(\pi_a, \pi_b\right)&:=   V(\pi_a)- V(\pi_b)=\sum_{j=0}^{m_{\max}} \underbrace{V_j(\pi_a)-V_j(\pi_b)}_{ :={\Delta}_j\left(\pi_a, \pi_b\right)}\\
 \end{aligned}$$
\end{definition}

By the above definitions, the regret of $\pihat^{\addDR}$ relative to $\piast$ can be bounded as
\begin{equation}\label{equ:pf:DRregret:first}
    \begin{aligned}
    R(\pihat^{\addDR}) &=V(\piast)-V(\pihat^{\addDR})\\
    &=V(\piast)-\widetilde{V}^{\addDR}(\piast)+\widetilde{V}^{\addDR}(\piast)-\widehat{V}^{\addDR}(\piast)+ \underbrace{\widehat{V}^{\addDR}(\piast)-\widehat{V}^{\addDR}(\pihat^{\addDR})}_{\leq 0,\text{ by definition of $\pihat^{\addDR}$}}\\
    & \quad + \widehat{V}^{\addDR}(\pihat^{\addDR})-\widetilde{V}^{\addDR}(\pihat^{\addDR})+\widetilde{V}^{\addDR}(\pihat^{\addDR})-V(\pihat^{\addDR})\\
    & \leq   \widetilde{\Delta}(\piast,\pihat^{\addDR})-\widehat{\Delta}(\piast,\pihat^{\addDR}) + \Delta(\piast,\pihat^{\addDR})-\widetilde{\Delta}(\piast,\pihat^{\addDR}) \\
    &\leq \sup_{\pi_a,\pi_b\in\Pi}\lvert  \widetilde{\Delta}(\pi_a,\pi_b)-\widehat{\Delta}(\pi_a,\pi_b)  \rvert + \sup_{\pi_a,\pi_b\in\Pi}\lvert \Delta(\pi_a,\pi_b)-\widetilde{\Delta}(\pi_a,\pi_b)\rvert.
    \end{aligned}
\end{equation}
In the following, we will bound each of these two terms separately.
\paragraph{Bounding $\sup_{\pi_a,\pi_b\in\Pi}\lvert  \widetilde{\Delta}(\pi_a,\pi_b)-\widehat{\Delta}(\pi_a,\pi_b)  \rvert$.}

Due to triangle inequality, we have that $ \sup_{\pi_a,\pi_b\in\Pi}\lvert  \widetilde{\Delta}(\pi_a,\pi_b)-\widehat{\Delta}(\pi_a,\pi_b)  \rvert\leq \sum_{j=0}^{m_{\max}}  \sup_{\pi_a,\pi_b\in\Pi}\lvert  \widetilde{\Delta}_j(\pi_a,\pi_b)-\widehat{\Delta}_j(\pi_a,\pi_b)  \rvert $.
Therefore, it suffices to bound $ \sup_{\pi_a,\pi_b\in\Pi}\lvert  \widetilde{\Delta}_j(\pi_a,\pi_b)-\widehat{\Delta}_j(\pi_a,\pi_b)  \rvert $.
We have the following decomposition:
 \begin{equation*}
        \begin{aligned}
      &    \sup_{\pi_a,\pi_b\in\Pi}\lvert  \widetilde{\Delta}_j(\pi_a,\pi_b)-\widehat{\Delta}_j(\pi_a,\pi_b)  \rvert 
      \\ 
      \leq &   \sup_{\pi_a,\pi_b\in\Pi} \Bigg\lvert\underbrace{\frac{1}{n}\sum_{i=1}^{n} \left[\left(\hat{\overline{\bg}}(\bX_i)-\overline{\bg}(\bX_i)\right)^\top \left(\phi(\bA_i)\phi(\bA_i)^\top\Sigma(\bX_i)^{-1} -1\right)\right]_j \left(\pi_a(\bX_{ij})-\pi_b(\bX_{ij})\right) }_{:=T_{1,j}(\pi_a,\pi_b)}\Bigg\rvert\\
            & +  \sup_{\pi_a,\pi_b\in\Pi}\Bigg\lvert \underbrace{\frac{1}{n}\sum_{i=1}^{n}   \left[\left(\overline{\bY}_{i}- \overline{\bg}^\top(\bX_i)\phi(\bA_i)\right)\phi(\bA_i)^\top \left( {\Sigma}(\bX_i)^{-1} -\widehat{\Sigma}(\bX_i)^{-1} \right)\right]_j \left(\pi_a(\bX_{ij})-\pi_b(\bX_{ij})\right)}_{:=T_{2,j}(\pi_a,\pi_b)} \Bigg\rvert\\
            & + \sup_{\pi_a,\pi_b\in\Pi} \Bigg\lvert \underbrace{\frac{1}{n}\sum_{i=1}^{n}\left[\left(\hat{\overline{\bg}}(\bX_i)-\overline{\bg}(\bX_i)\right)^\top \phi(\bA_i)\phi(\bA_i)^\top  \left( \widehat{\Sigma}(\bX_i)^{-1} -{\Sigma}(\bX_i)^{-1} \right)\right]_j\left(\pi_a(\bX_{ij})-\pi_b(\bX_{ij})\right)}_{:=T_{3,j}(\pi_a,\pi_b)} \Bigg\rvert.      
        \end{aligned}
    \end{equation*}
We bound each of the three terms in turn.
Since the nuisance models are trained on an independent data split, we have shown in the proof of Theorem~\ref{thm:semieff} that term $T_{1,j}(\pi_a,\pi_b)$ has mean 0 for any fixed $\pi_a,\pi_b\in\Pi$. Therefore, we can write the following:
\[
\begin{aligned}
 \sup_{\pi_a,\pi_b\in\Pi}\lvert T_{1,j}(\pi_a,\pi_b)\rvert & = \sup_{\pi_a,\pi_b\in\Pi}\left\lvert T_{1,j}(\pi_a,\pi_b) -\E[ T_{1,j}(\pi_a,\pi_b)]\right\rvert.
\end{aligned}
\]
Next, by Assumption~\ref{ass:rate}, we know that
$$
 \sup_{\bx}\left\{\left|\hat{\overline{g}}^{(j)}(\bx)-{\overline{g}}^{(j)}(\bx)\right|\right\} \xrightarrow{p} 0\quad \text{for}\ j=0,\ldots,m
$$
with probability tending to 1. Consequently, the individual summands in $T_{1,j}(\pi)$
is bounded with probability tending to 1. 
Combining this with the fact that clusters are i.i.d., we can leverage Lemma~2 of \cite{zhou2023offline} by treating the summands in $T_{1,j}(\pi)$ as the i.i.d. $\Gamma_i$ in their result, and specializing it to the one-dimensional case. Thus, we obtain the result: $\forall \delta>0$, with probability at least $1-2 \delta$,
\begin{equation}\label{pf:equ:T1}
    \begin{aligned}
    & \sup_{\pi_a,\pi_b\in\Pi}\left\lvert T_{1,j}(\pi_a,\pi_b) -\E[ T_{1,j}(\pi_a,\pi_b)]\right\rvert \\
    \leq &\  o\left(\frac{1}{\sqrt{n}}\right) + \left(54.4 \sqrt{2} \kappa(\Pi)+435.2+\sqrt{2 \log \frac{1}{\delta}}\right) \sqrt{\frac{V_*}{n}},
    \end{aligned}
\end{equation}
where $V_*$ is defined as the worst-case variance of evaluating
the difference between two policies in $\Pi$:
\begin{equation}        
V_*:=\sup_{\pi_a,\pi_b\in\Pi}\E\left[ \left[\left(\hat{\overline{\bg}}(\bX_i)-\overline{\bg}(\bX_i)\right)^\top \left(\phi(\bA_i)\phi(\bA_i)^\top\Sigma(\bX_i)^{-1} -1\right)\right]^2_j \left(\pi_a(\bX_{ij})-\pi_b(\bX_{ij})\right)^2 \right].
\end{equation} 
Since the propensity scores are strictly bounded away from zero, the propensity score matrix $\Sigma(\bX_i)^{-1} $ is bounded, and therefore, under Assumption~\ref{ass:rate}, $V_*$ can be upper bounded by
\begin{equation}
    \begin{aligned}
    V_* & \leq \E\left[ \left[\left(\hat{\overline{\bg}}(\bX_i)-\overline{\bg}(\bX_i)\right)^\top \left(\phi(\bA_i)\phi(\bA_i)^\top\Sigma(\bX_i)^{-1} -1\right)\right]^2_j  \right] \\
    & \lesssim \E\left[\left(\hat{\overline{\bg}}(\bX_i)-\overline{\bg}(\bX_i)\right)^\top\left(\hat{\overline{\bg}}(\bX_i)-\overline{\bg}(\bX_i)\right) \right]\\
    & =  O_{p}\left(r_{n,g}^2\right) =o_p(1).
   \end{aligned}
\end{equation}
Combining this with Equation~\eqref{pf:equ:T1}, we have
\begin{equation}\label{pf:equ:T1:bound}
    \begin{aligned}
        & \sup_{\pi_a,\pi_b\in\Pi}\left\lvert T_{1,j}(\pi_a,\pi_b)\right\rvert 
        \leq  o\left(\frac{1}{\sqrt{n}}\right) + o_p\left(\frac{1}{\sqrt{n}}\right)= o_p\left(\frac{1}{\sqrt{n}}\right).
    \end{aligned}
\end{equation}
Following the similar argument, we have $\sup_{\pi_a,\pi_b\in\Pi}\left\lvert T_{2,j}(\pi_a,\pi_b)\right\rvert = o_p\left(\frac{1}{\sqrt{n}}\right)$.
Next, we bound $\sup_{\pi_a,\pi_b\in\Pi}\left\lvert T_{3,j}(\pi_a,\pi_b)\right\rvert$ as follows:
\begin{equation}
    \begin{aligned}
    & \sup_{\pi_a,\pi_b\in\Pi}\left\lvert T_{3,j}(\pi_a,\pi_b)\right\rvert\\
      =  & \sup_{\pi_a,\pi_b\in\Pi} \Bigg\lvert \frac{1}{n}\sum_{i=1}^{n}\left[\left(\hat{\overline{\bg}}(\bX_i)-\overline{\bg}(\bX_i)\right)^\top \phi(\bA_i)\phi(\bA_i)^\top  \left( \widehat{\Sigma}(\bX_i)^{-1} -{\Sigma}(\bX_i)^{-1} \right)\right]_j\left(\pi_a(\bX_{ij})-\pi_b(\bX_{ij})\right) \Bigg\rvert \\
     \leq   & \ \frac{1}{n}\Bigg\lvert\left[\left(\hat{\overline{\bg}}(\bX_i)-\overline{\bg}(\bX_i)\right)^\top \phi(\bA_i)\phi(\bA_i)^\top  \left( \widehat{\Sigma}(\bX_i)^{-1} -{\Sigma}(\bX_i)^{-1} \right)\right]_j\Bigg\rvert  \\
     = & \ \frac{1}{n}\sum_{i=1}^{n}\Bigg\lvert\left(\hat{\overline{\bg}}(\bX_i)-\overline{\bg}(\bX_i)\right)^\top \phi(\bA_i)\phi(\bA_i)^\top  \left( \widehat{\Sigma}(\bX_i)^{-1} -{\Sigma}(\bX_i)^{-1} \right) \boldsymbol{v}_j\Bigg\rvert\\
     \leq & \sqrt{\frac{1}{n}\sum_{i=1}^{n} \left(\left(\hat{\overline{\bg}}(\bX_i)-\overline{\bg}(\bX_i)\right)^\top \phi(\bA_i)\right)^2} \sqrt{\frac{1}{n}\sum_{i=1}^{n}\left(\phi(\bA_i)^\top  \left( \widehat{\Sigma}(\bX_i)^{-1} -{\Sigma}(\bX_i)^{-1} \right) \boldsymbol{v}_j\right)^2 },
    \end{aligned}
\end{equation}
where $\boldsymbol{v}_j$ is a standard basis vector with 1 in the $j$th position and 0 elsewhere, and the last inequality above follows from Cauchy-Schwartz.
Taking the expectation of both sides yields:
\begin{equation*}
    \begin{aligned}
    & \E\left[\sup_{\pi_a,\pi_b\in\Pi}\left\lvert T_{3,j} (\pi_a,\pi_b)\right\rvert\right] \\
    \leq & \ \E\left[ \sqrt{\frac{1}{n}\sum_{i=1}^{n} \left(\left(\hat{\overline{\bg}}(\bX_i)-\overline{\bg}(\bX_i)\right)^\top \phi(\bA_i)\right)^2} \sqrt{\frac{1}{n}\sum_{i=1}^{n}\left(\phi(\bA_i)^\top  \left( \widehat{\Sigma}(\bX_i)^{-1} -{\Sigma}(\bX_i)^{-1} \right) \boldsymbol{v}_j\right)^2 }\right] \\
    \leq & \ \sqrt{\E\left[ \frac{1}{n}\sum_{i=1}^{n}\left(\left(\hat{\overline{\bg}}(\bX_i)-\overline{\bg}(\bX_i)\right)^\top\phi(\bA_i)\right)^2\right]}\times  \sqrt{\E\left[ \frac{1}{n}\sum_{i=1}^{n} \left( \phi(\bA_i)^\top  \left( \widehat{\Sigma}(\bX_i)^{-1} -{\Sigma}(\bX_i)^{-1} \right)\boldsymbol{v}_j)  \right)^2\right] }\\
    = & \ \sqrt{\E\left[\left(\left(\hat{\overline{\bg}}(\bX_i)-\overline{\bg}(\bX_i)\right)^\top\phi(\bA_i)\right)^2\right]}\times  \sqrt{\E\left[ \left( \phi(\bA_i)^\top  \left( \widehat{\Sigma}(\bX_i)^{-1} -{\Sigma}(\bX_i)^{-1} \right)\boldsymbol{v}_j \right)^2\right] }\\
    =  & \sqrt{O_p(r_{n,g}^2)} \sqrt{O_p(r_{n,\Sigma}^2)} = o_p\left(\frac{1}{\sqrt{n}}\right).
    \end{aligned}
\end{equation*} 
Consequently, by Markov's inequality, this immediately implies $\sup_{\pi_a,\pi_b\in\Pi}\left\lvert T_{3,j}(\pi_a,\pi_b)\right\rvert=o_p\left(\frac{1}{\sqrt{n}}\right)$.
 Putting the above bounds for $\sup_{\pi_a,\pi_b\in\Pi}\left\lvert T_{1,j}(\pi_a,\pi_b)\right\rvert,\sup_{\pi_a,\pi_b\in\Pi}\left\lvert T_{2,j}(\pi_a,\pi_b)\right\rvert$ and $\sup_{\pi_a,\pi_b\in\Pi}\left\lvert T_{3,j}(\pi_a,\pi_b)\right\rvert$ together, we arrive at $$\sup_{\pi_a,\pi_b\in\Pi}\lvert  \widetilde{\Delta}(\pi_a,\pi_b)-\widehat{\Delta}(\pi_a,\pi_b) \rvert=o_p\left(\frac{1}{\sqrt{n}}\right).$$

\paragraph{Bounding $ \sup_{\pi_a,\pi_b\in\Pi}\lvert \Delta(\pi_a,\pi_b)-\widetilde{\Delta}(\pi_a,\pi_b)\rvert$.} 
By triangle inequality, we have that $ \sup_{\pi_a,\pi_b\in\Pi}\lvert  \widetilde{\Delta}(\pi_a,\pi_b)-{\Delta}(\pi_a,\pi_b)  \rvert\leq \sum_{j=1}^{m_{\max}}  \sup_{\pi_a,\pi_b\in\Pi}\lvert  \widetilde{\Delta}_j(\pi_a,\pi_b)-{\Delta}_j(\pi_a,\pi_b)  \rvert $. (Note that we exclude the index $j=0$  in the summation since $\widetilde{\Delta}_0(\pi_a,\pi_b),{\Delta}_0(\pi_a,\pi_b)$ are always zero given that $\pi(\bX_{i0})\equiv1$ for any $\pi\in\Pi$ by definition, and thus their difference is also zero.)
Therefore, it suffices to bound $ \sup_{\pi_a,\pi_b\in\Pi}\lvert  \widetilde{\Delta}_j(\pi_a,\pi_b)-{\Delta}_j(\pi_a,\pi_b)\rvert$ for $j=1,\ldots,m_{\max}$.

We have the following: 
\begin{equation}
    \begin{aligned}
       & \sup_{\pi_a,\pi_b\in\Pi}\lvert  \widetilde{\Delta}_j(\pi_a,\pi_b)-{\Delta}_j(\pi_a,\pi_b)\rvert \\ 
      = &  \sup_{\pi_a,\pi_b\in\Pi}\left\lvert \frac{1}{n}\sum_{i=1}^{n}\left[\ {\overline{\bg}}(\bX_i) +  {\Sigma}(\bX_i)^{-1}\left(\overline{\bY}_{i}- {\overline{\bg}}^\top(\bX_i)\phi(\bA_i)\right)\phi(\bA_i) \right]_j \left(\pi_a(\bX_{ij})-\pi_b(\bX_{ij})\right)  \right.\\
      & \qquad \qquad - \left. \E\left[\left[\ {\overline{\bg}}(\bX_i) +  {\Sigma}(\bX_i)^{-1}\left(\overline{\bY}_{i}- {\overline{\bg}}^\top(\bX_i)\phi(\bA_i)\right)\phi(\bA_i) \right]_j \left(\pi_a(\bX_{ij})-\pi_b(\bX_{ij})\right)  \right]\right\rvert \\ 
    \end{aligned}
\end{equation}
By Assumptions~\ref{ass:iidcluster},~\ref{ass:DGP}~and~\ref{ass:regret:bound}(a),  the summands in $\widetilde{\Delta}_j(\pi_a,\pi_b)$ are i.i.d. with bounded support.
Therefore, we can again apply Lemma 2 of \cite{zhou2023offline} to obtain the result:
$\forall \delta>0$, with probability at least $1-2 \delta$,
\begin{equation}
    \begin{aligned}
        & \sup_{\pi_a,\pi_b\in\Pi}\lvert  \widetilde{\Delta}_j(\pi_a,\pi_b)-{\Delta}_j(\pi_a,\pi_b)\rvert  \\
        \leq &\  o\left(\frac{1}{\sqrt{n}}\right) + \left(54.4 \sqrt{2} \kappa(\Pi)+435.2+\sqrt{2 \log \frac{1}{\delta}}\right) \sqrt{\frac{V_{*}^j}{n}},
    \end{aligned}
\end{equation}
where $V_*^j$ is defined as:
\begin{equation}
    \begin{aligned}
        V_*^j := & \sup_{\pi_a,\pi_b\in\Pi}\E\left[\left[\ {\overline{\bg}}(\bX_i) +  {\Sigma}(\bX_i)^{-1}\left(\overline{\bY}_{i}- {\overline{\bg}}^\top(\bX_i)\phi(\bA_i)\right)\phi(\bA_i) \right]^2_j\left(\pi_a(\bX_{ij})-\pi_b(\bX_{ij})\right)^2 \right]\\
        \leq & \ \E\left[\left[\ {\overline{\bg}}(\bX_i) +  {\Sigma}(\bX_i)^{-1}\left(\overline{\bY}_{i}- {\overline{\bg}}^\top(\bX_i)\phi(\bA_i)\right)\phi(\bA_i) \right]^2_j \right],
    \end{aligned}
\end{equation} 
which is bounded by the second moment of the $j$th element of the doubly robust score vector.
Summing over $j$, with probability at least $1-2m_{\max}\delta$, we have:
\begin{equation}
    \begin{aligned}
        & \sup_{\pi_a,\pi_b\in\Pi}\lvert \Delta(\pi_a,\pi_b)-\widetilde{\Delta}(\pi_a,\pi_b)\rvert 
        \leq o\left(\frac{1}{\sqrt{n}}\right)+ \left(54.4 \sqrt{2} \kappa(\Pi)+435.2+\sqrt{2 \log \frac{1}{\delta}}\right) \frac{\sum_{j=1}^{m_{\max}}\sqrt{V_{*}^j}}{\sqrt{n}}.
    \end{aligned}
\end{equation}
We also note that, in our case, $\Pi$ is a VC class with finite VC dimension $\nu$. By Theorem 1 of \cite{haussler1995sphere}, the covering number can be bounded by VC dimension as follows: $N_H(\epsilon, \Pi) \leq e(V C(\Pi)+1)\left(\frac{2 e}{\epsilon}\right)^{V C(\Pi)}$. By taking the natural log of both sides and computing the entropy integral, one can show that $\kappa\left(\Pi\right) \leq 2.5 \sqrt{V C\left(\Pi\right)}=2.5\sqrt{\nu}$. Then, plugging in the upper bound on $\kappa(\Pi)$, we can show that 
$\sup_{\pi_a,\pi_b\in\Pi}\lvert \Delta(\pi_a,\pi_b)-\widetilde{\Delta}(\pi_a,\pi_b)\rvert 
\leq O_p(\sqrt{\frac{\nu}{n}})$.

Combining the two bounds above yields the desired result.
\end{proof}

\section{Additional results for observational studies}

\subsection{Plug-in approach with unknown propensity scores}\label{appendix:unknownPS}

We now consider a plug-in approach to observational studies where propensity scores are unknown and must be estimated.
As shown in \cite{kitagawa2018should}, if we have consistent estimates $\hat{e}_j(a\mid\bx)$ of the propensities and plug them into Equation~\eqref{equ:Vhat:addIPW}, we can still learn a consistent policy $\pihat$ by maximizing $\widehat{V}(\pi)$. 
We define the estimated policy based on this plug-in approach as,
\begin{equation}\label{equ:pihat:plugin}
    \pihat_{\hat{e}}\ := \  \underset{\pi\in\Pi}{\operatorname{argmax}}\   \widehat{V}_{\hat{e}} \quad \text{where} \quad \widehat{V}_{\hat{e}}\ := \
    \frac{1}{n}\sum_{i=1}^{n}\overline{\bY}_{i} \left\{\sum_{j=1}^{M_i}
  \left( \frac{\mathds{1}\left\{
       A_{ij}=\pi(\bX_{ij})
     \right\}}{  \hat{e}_j(
 \pi(\bX_{ij})
      \mid\bX_{i})}
  -1\right)+1\right\},
\end{equation}
where $\hat{e}_j$ is the estimated propensity score.
We examine how the estimation of propensity scores affect the performance of learned policy.
The next assumption concerns the convergence rate of the estimation error of the propensity score.
\begin{assumption}[Estimation error of $e_j$]\label{ass:propensity:error}
For some sequence $\rho_n\rightarrow\infty$, assume that 
$$ \mathbb{E}\left[\frac{1}{n}\sum_{i=1}^{n}\sup_{a\in\{0,1\}}\left|\sum_{j=1}^{M_i} \frac{1}{  {e}_j(
 a
      \mid\bX_{i})}
  - \frac{1}{  \hat{e}_j(
a
      \mid\bX_{i})} \right|\right]=O\left(\rho_n^{-1}\right).$$
\end{assumption}

We next derive the regret guarantees when using the estimated propensity score.
Theorem~\ref{thm:regret:plugin} shows that a simple plug-in approach based on the IPW-type estimator in Equation~\eqref{equ:Vhat:addIPW} typically leads to a learned policy with a slower-than-$\sqrt{n}$ convergence rate.
\begin{theorem}[Regret bound with estimated propensity score]\label{thm:regret:plugin}
    Suppose Assumptions~\ref{ass:iidcluster}--\ref{ass:regret:bound} and \ref{ass:propensity:error} hold. Define $ \pihat_{\hat{e}}$ as the solution in Equation~\eqref{equ:pihat:plugin}.
The regret of $ \pihat_{\hat{e}}$ can be upper bounded as,
\begin{equation}
    R( \pihat_{\hat{e}})\leq O_p(n^{-\frac{1}{2}}\vee\rho_n^{-1}).
\end{equation}
\end{theorem}
\begin{proof}
By definition, the regret of $\pihat_{\hat{e}}$ relative to $\piast$ is given by
\begin{equation*}
    \begin{aligned}
    R(\pihat_{\hat{e}}) & = V(\piast)-V(\pihat_{\hat{e}}) \\
    & = V(\piast)-\widehat{V}(\piast)+\widehat{V}(\piast)-\widehat{V}_{\hat{e}}(\piast)+\widehat{V}_{\hat{e}}(\piast)-\widehat{V}_{\hat{e}}(\pihat_{\hat{e}})+\widehat{V}_{\hat{e}}(\pihat_{\hat{e}})-\widehat{V}(\pihat_{\hat{e}})+\widehat{V}(\pihat_{\hat{e}})-V(\pihat_{\hat{e}})\\
    &\leq 2\sup_{\pi\in\Pi}\lvert \widehat{V}(\pi)-V(\pi) \rvert + 2\sup_{\pi\in\Pi}\lvert \widehat{V}(\pi)-\widehat{V}_{\hat{e}}(\pi) \rvert ,
        \end{aligned}
    \end{equation*}
where $\widehat{V}(\cdot)$ is the empirical policy value estimate using the true propensity score in  Equation~\eqref{equ:Vhat:addIPW}.
    Due to Theorem~\ref{thm:regret}, the first term is bounded by $O_p(n^{-\frac{1}{2}})$ . We now study the second term.
\begin{equation*}
    \begin{aligned}
        \sup_{\pi\in\Pi}\lvert \widehat{V}(\pi)-\widehat{V}_{\hat{e}}(\pi) \rvert &=\sup_{\pi\in\Pi}\left|\frac{1}{n} \sum_{i=1}^n \overline{Y}_{i} \sum_{j=1}^{M_i}
\left(\frac{\mathds{1}\left\{
       A_{ij}=\pi(\bX_{ij})
     \right\}}{  {e}_j(
 \pi(\bX_{ij})
      \mid\bX_{i})}
  - \frac{\mathds{1}\left\{
       A_{ij}=\pi(\bX_{ij})
     \right\}}{  \hat{e}_j(
 \pi(\bX_{ij})
      \mid\bX_{i})}\right) \right|\\
      &\leq B  \frac{1}{n} \sum_{i=1}^n \sup_{a\in\{0,1\}}\left|\sum_{j=1}^{M_i} \frac{1}{  {e}_j(
a
      \mid\bX_{i})}
  - \frac{1}{  {e}_j( a
      \mid\bX_{i})} \right|\\
      &=O_p(\rho_n^{-1}).
    \end{aligned}
\end{equation*}
Combining the above, we obtain the desired result.
\end{proof}

\subsection{Mixed-integer linear program formulation for the doubly robust estimator}\label{appendix:MILP:DR}

Similar to Section~\ref{sec:MIP}, we consider linear policy rules as an example.
By introducing binary variables $p_{ij}=\pi(\bX_{ij})$, we can rewrite the doubly robust estimator in Equation~\eqref{equ:DRestor:closedform} (up to some constants) as,
\[
 \frac{1}{n}\sum_{i=1}^{n}  \sum_{j=1}^{M_i}
 \left\{ \hat{\overline{g}}^{(j)}(\bX_i)+ \left(\overline{\bY}_{i}- \hat{\overline{\bg}}(\bX_i)^\top\phi(\bA_i)\right) \frac{A_{ij}- \hat{e}_j(1 \mid\bX_{i})}{  \hat{e}_j(1 \mid\bX_{i}) \hat{e}_j(0 \mid\bX_{i})} \right\} p_{ij}.
\]
This implies that solving $\pihat^{\addDR}$ can be equivalently represented as solving the following linear MIP:
\begin{equation}
    \begin{aligned}
&\max_{\substack{\bbeta \in \mathcal{B}, \{p_{ij}\} \in \mathbb{R}}}  \; \frac{1}{n}\sum_{i=1}^{n}  \sum_{j=1}^{M_i}
 \left\{\hat{\overline{g}}^{(j)}(\bX_i)+ \left(\overline{\bY}_{i}- \hat{\overline{\bg}}(\bX_i)^\top\phi(\bA_i)\right) \frac{A_{ij}- \hat{e}_j(1 \mid\bX_{i})}{  \hat{e}_j(1 \mid\bX_{i}) \hat{e}_j(0 \mid\bX_{i})} \right\} p_{ij}\\
&\text { s.t. } 
\begin{aligned}
   & \frac{\bX_{ij}^\top \bbeta}{C_{ij}}<   p_{ij} \leq 1+\frac{\bX_{ij}^\top\bbeta}{C_{ij}} \quad \text { for } i=1, \ldots, n, \quad\text{and }\quad j=1,\ldots,M_i,\\
& p_{ij} \in\{0,1\},
\end{aligned}
\end{aligned}
\end{equation}
where constants $C_{ij}$ should satisfy $C_{ij}>\sup _{\bbeta \in \mathcal{B}}\left|\bX_{ij}^\top \bbeta\right|$. 

\section{The estimator for the general polynomial additive model} \label{app:higherorder}

Here, we give an explicit formula for the policy estimator under the general polynomial additive structural model given in Equation~\eqref{equ:polyadditive}.
We first index the entries of $\mathbb{E}\left[\phi(\bA_i)\phi(\bA_i)^\top \mid \bX_i\right]$ by the sets $\mathcal{J}$ and $\mathcal{K}$ corresponding to the matrix row and column, where $\mathcal{J}$ and $\mathcal{K}$ contain the indices of individuals associated with the respective entries of $\phi(\bA_i)$.
According to factored propensity score, we have the following decomposition
\[
\left(\mathbb{E}\left[\phi(\bA_i)\phi(\bA_i)^\top \mid \bX_i\right]\right)_{\mathcal{J}, \mathcal{K}}=\prod_{j \in \mathcal{J}\cup \mathcal{K}}e_j(\bX_i).
\]
We let $\mathcal{I}_i^\beta$ denote the power set of $\{\emptyset,1,\ldots,M_i\}$ with cardinality at most $\beta$.
For simplicity, we define $e_j\left(\bX_i\right):=e_j\left(1 \mid \bX_i\right)$.

Now, we provide an explicit expression for the inverse of $\mathbb{E}\left[\phi(\bA_i)\phi(\bA_i)^\top \mid \bX_i\right]$ by citing the results of Lemma 1 of \cite{cortez2022exploiting}:
\begin{lemma}\label{lemma:inverseMat}
    The matrix $\Sigma(\bX_i):=\mathbb{E}\left[\phi(\bA_i)\phi(\bA_i)^\top \mid \bX_i\right]$ is invertible, with each entry of its inverse matrix $\Sigma^{-1}(\bX_i)$ given by the following formula 
    \[
    \left(\Sigma^{-1}(\bX_i)\right)_{\mathcal{J},\mathcal{K}}=\prod_{j \in \mathcal{J}} \frac{-1}{e_j(\bX_i)} \prod_{k \in \mathcal{K}} \frac{-1}{e_k(\bX_i)} \sum_{\substack{\mathcal{U} \in \mathcal{I}_i^\beta \\(\mathcal{J} \cup \mathcal{K}) \subseteq \mathcal{U}}} \prod_{\ell \in \mathcal{U}} \frac{e_{\ell}(\bX_i)}{1-e_{\ell}(\bX_i)}.
    \]
\end{lemma}
Recalling Equation~\eqref{equ:polyadditive}, the generalized form of the estimator is given by
\begin{equation*}
    \begin{aligned}
    \widehat{V}(\pi)
    &=\frac{1}{n}\sum_{i=1}^{n} \overline{Y}_{i}\cdot\phi({\pi} (\bX_i))^\top\mathbb{E}\left[\phi(\bA_i)\phi(\bA_i)^\top \mid \bX_i\right]^{-1} \phi(\bA_i),
    \end{aligned}
\end{equation*}
Based on Lemma~\ref{lemma:inverseMat}, we can directly calculate the cluster-level weights, resulting into
\begin{align*}
    \phi({\pi}(\bX_i))^\top\Sigma^{-1}(\bX_i)\phi(\bA_i)&= \sum_{\mathcal{J} \in \mathcal{I}_i^\beta } \sum_{\mathcal{K} \in \mathcal{I}_i^\beta } \prod_{j \in \mathcal{J}} \frac{-A_{ij}}{e_j(\bX_i)} \prod_{k \in \mathcal{K}} \frac{-\pi(\bX_{ik})}{e_k(\bX_i)} \sum_{\substack{\mathcal{U} \in \mathcal{I}_i^\beta \\(\mathcal{J} \cup \mathcal{K}) \subseteq \mathcal{U}}} \prod_{\ell \in \mathcal{U}} \frac{e_{\ell}(\bX_i)}{1-e_{\ell}(\bX_i)}\\
    & = \sum_{\mathcal{J} \in \mathcal{I}_i^\beta }\prod_{j \in \mathcal{J}} \frac{-A_{ij}}{e_j(\bX_i)}
    \sum_{\mathcal{K} \in \mathcal{I}_i^\beta }  \prod_{k \in \mathcal{K}} \frac{-\pi(\bX_{ik})}{e_k(\bX_i)} 
    \sum_{\substack{\mathcal{U} \in \mathcal{I}_i^\beta \\(\mathcal{J} \cup \mathcal{K}) \subseteq \mathcal{U}}} \prod_{\ell \in \mathcal{U}} \frac{e_{\ell}(\bX_i)}{1-e_{\ell}(\bX_i)}\\
    & = \sum_{\mathcal{J} \in \mathcal{I}_i^\beta }\prod_{j \in \mathcal{J}} \frac{-A_{ij}}{e_j(\bX_i)} \sum_{\substack{\mathcal{U} \in \mathcal{I}_i^\beta \\\mathcal{J}  \subseteq \mathcal{U}}} \prod_{\ell \in \mathcal{U}} \frac{e_{\ell}(\bX_i)}{1-e_{\ell}(\bX_i)}
     \sum_{\mathcal{K}  \subseteq \mathcal{U}}  \prod_{k \in \mathcal{K}} \frac{-\pi(\bX_{ik})}{e_k(\bX_i)}\\
     & = \sum_{\mathcal{J} \in \mathcal{I}_i^\beta }\prod_{j \in \mathcal{J}} \frac{-A_{ij}}{e_j(\bX_i)} 
     \sum_{\substack{\mathcal{U} \in \mathcal{I}_i^\beta \\\mathcal{J}  \subseteq \mathcal{U}}} \prod_{\ell \in \mathcal{U}} \frac{e_{\ell}(\bX_i)}{1-e_{\ell}(\bX_i)}
      \prod_{k \in \mathcal{U}} \left(1- \frac{\pi(\bX_{ik})}{e_k(\bX_i)}\right)\\
    & = \sum_{\mathcal{U} \in \mathcal{I}_i^\beta }\prod_{\ell \in \mathcal{U}} \frac{e_{\ell}(\bX_i)}{1-e_{\ell}(\bX_i)} 
        \prod_{k \in \mathcal{U}} \left(1- \frac{\pi(\bX_{ik})}{e_k(\bX_i)}\right)
     \sum_{\mathcal{J}  \subseteq \mathcal{U}}  \prod_{j \in \mathcal{J}}
     \frac{-A_{ij}}{e_j(\bX_i)} \\
     & =  \sum_{\mathcal{U} \in \mathcal{I}_i^\beta }\prod_{\ell \in \mathcal{U}} \frac{e_{\ell}(\bX_i)}{1-e_{\ell}(\bX_i)} 
        \prod_{k \in \mathcal{U}} \left(1- \frac{\pi(\bX_{ik})}{e_k(\bX_i)}\right)
     \prod_{j \in \mathcal{U}}
    \left(1- \frac{A_{ij}}{e_j(\bX_i)} \right)\\
    & =  \sum_{\mathcal{U} \in \mathcal{I}_i^\beta }\prod_{\ell \in \mathcal{U}}
    \frac{(e_{\ell}(\bX_i)-A_{i\ell})(e_{\ell}(\bX_i)-\pi(\bX_{i\ell}))}{e_{\ell}(\bX_i)(1-e_{\ell}(\bX_i))}\\
    & =  \sum_{\mathcal{U} \in \mathcal{I}_i^\beta }\prod_{\ell \in \mathcal{U}}
    \frac{(e_{\ell}(\bX_i)-A_{i\ell})(e_{\ell}(\bX_i)-\pi(\bX_{i\ell}))}{e_{\ell}(\bX_i)(1-e_{\ell}(\bX_i))}\\
    & = \sum_{\mathcal{U} \in \mathcal{I}_i^\beta }\prod_{\ell \in \mathcal{U}} 
   \left( \frac{\mathds{1}\left\{
       A_{i\ell}=\pi(\bX_{i\ell})
     \right\}}{  e_\ell(
 \pi(\bX_{i\ell})
      \mid \bX_{i})}
  -1\right).
\end{align*}
Consequently, the explicit form of the final policy value estimator is given by:
\begin{equation*}
    \begin{aligned}
    \widehat{V}(\pi)
    &=\frac{1}{n}\sum_{i=1}^{n} \overline{Y}_{i}\sum_{\mathcal{U} \in \mathcal{I}_i^\beta }\prod_{\ell \in \mathcal{U}} 
   \left( \frac{\mathds{1}\left\{
       A_{i\ell}=\pi(\bX_{i\ell})
     \right\}}{  e_\ell(
 \pi(\bX_{i\ell})
      \mid \bX_{i})}
  -1\right).
    \end{aligned}
\end{equation*}

\end{document}